\documentclass[12pt,oneside,reqno]{article}
\makeatletter
\newcommand{\singlespacing}{\let\CS=\@currsize\renewcommand{\baselinestretch}{1}\tiny\CS}
 \usepackage{a4wide}
 \usepackage[a4paper, total={6.4in, 9.5in}]{geometry}
\usepackage{epsfig}
\usepackage{lineno}
\usepackage{rotating}
\usepackage{booktabs}
\usepackage{wrapfig}
\usepackage{comment}
\usepackage{multirow}
\usepackage{siunitx}
\usepackage{rotating}
\usepackage[T1]{fontenc}
\usepackage[utf8]{inputenc}
\usepackage{amsmath}
\usepackage{amsthm}

\usepackage{pdfpages}
\usepackage{amsfonts,graphicx}
\usepackage{authblk}
\usepackage{verbatim}
\usepackage[mathscr]{eucal}
\usepackage{booktabs}
\usepackage{siunitx}

\usepackage{mathptmx} 
\usepackage{adjustbox}
\usepackage{xcolor}
\usepackage{xspace}
\usepackage{enumitem}
\usepackage{subfig}
\usepackage[colorlinks=true,urlcolor=blue,citecolor=blue,linkcolor=blue,bookmarks=true]{hyperref}

\usepackage{cleveref}
\crefname{appendix}{Appendix}{Appendices}
\usepackage[numbers,sort&compress]{natbib}
\usepackage{float}
\usepackage{epstopdf}
\usepackage[parfill]{parskip}

\newcommand{\RN}[1]{%
	\textup{\uppercase\expandafter{\romannumeral#1}}%
}

\newcommand{\be}{\begin{equation}}
\newcommand{\ee}{\end{equation}}
\newcommand{\beanno}{\begin{eqnarray*}}

\newcommand{\eeanno}{\end{eqnarray*}}
\newcommand{\bea}{\begin{eqnarray}}
\newcommand{\eea}{\end{eqnarray}}
\newcommand{\ba}{\begin{array}}
\newcommand{\ea}{\end{array}}

\newcommand{\bc}{\begin{center}}
\newcommand{\ec}{\end{center}}

\newcommand{\ncom}{\newcommand}

\ncom{\half}{\frac{1}{2}}
\ncom{\mbx}{\makebox{.25cm}}
\ncom{\hs}{\mbox{\hspace{.25cm}}}
\ncom{\rar}{\rightarrow}
\ncom{\Rar}{\Rightarrow}
\ncom{\noin}{\noindent}
\ncom{\sz}{\scriptsize}
\ncom{\rf}{\ref}
\ncom{\s}{\sqrt{2}}
\ncom{\sgm}{\sigma}
\ncom{\Sgm}{\Sigma}
\ncom{\psgm}{\sigma^{\prime}}
\ncom{\dt}{\delta}
\ncom{\Dt}{\Delta}
\ncom{\lmd}{\la}
\ncom{\Lmd}{\la}
\ncom{\Th}{\Theta}
\ncom{\e}{\xi}
\ncom{\eps}{\epsilon}
\ncom{\pcc}{\stackrel{P}{>}}
\ncom{\lp}{\stackrel{L_{p}}{>}}
\ncom{\dist}{{\rm\,dist}}
\ncom{\sspan}{{\rm\,span}}
\ncom{\re}{{\rm Re\,}}
\ncom{\im}{{\rm Im\,}}
\ncom{\sgn}{{\rm sgn\,}}
\clearpage
\ncom{\hone}{\mbox{\hspace{1em}}}
\ncom{\htwo}{\mbox{\hspace{2em}}}
\ncom{\hthree}{\mbox{\hspace{3em}}}
\ncom{\hfour}{\mbox{\hspace{4em}}}
\ncom{\vone}{\vskip 2ex}
\ncom{\vtwo}{\vskip 4ex}
\ncom{\vonee}{\vskip 1.5ex}
\ncom{\vthree}{\vskip 6ex}
\ncom{\vfour}{\vspace*{8ex}}
\ncom{\norm}{\|\;\;\|}
\ncom{\integ}[4]{\int_{#1}^{#2}\,{#3}\,d{#4}}
\ncom{\vspan}[1]{{{\rm\,span}\{ #1 \}}}
\ncom{\dm}[1]{ {\displaystyle{#1} } }
\ncom{\ri}[1]{{#1} \index{#1}}
\newtheorem{theorem}{\bf Theorem}[section]

\newtheorem{corollary}{Corollary}[section]

\newtheoremstyle
{remarkstyle}
{}
{11pt}
{}
{}
{\bfseries}
{:}
{     }
{\thmname{#1} \thmnumber{#2} }
\theoremstyle{remarkstyle}

\renewcommand{\baselinestretch}{1.5}
\def \x {x}
\def \X {X}
\def \y {y}

\def \a {a}
\def \b {b}
\def \th {\theta}

\def \f {f}
\def \F {F}
\def \ss {n}
\def \i {i}
\def \fx {\phi(\x;\th)}
\def \fxi {\phi(\x_i; \th)}
\def \fy {\phi(\y;\th)}

\def \pv {\Theta}
\def \fm {\phi(\mu;\theta)}
\def \fmt {\phi(\mu_{i};\theta)}
\def \fy {\phi(\y;\theta)}
\def \fyt {\phi(\y_i;\theta)}

\allowdisplaybreaks
\begin{document}
	\title{\bf Kumaraswamy–Teissier Distribution: Bayesian and Non-Bayesian Inference, Regression Modeling and Climate Application}
	\author{ Kamana Mishra$^1$, Neeraj Poonia$^2$, Tanmay Kayal$^1$ and Sarita Azad$^1$ \thanks{Corresponding Author Email: sarita@iitmandi.ac.in}}
	\affil{\it $^1$School of Mathematical and Statistical Sciences \\ \it Indian Institute of Technology Mandi, Kamand, Himachal Pradesh-175005, India  \\ $^2$ Cytel Innovation Team, Cytel Pune, Maharashtra-411038, India}
	\date{}
	\maketitle
	\vspace{-1.5em}
	\begin{abstract}
		\noindent This study introduces the three-parameter Kumaraswamy--Teissier distribution (KTD) by applying the Kumaraswamy-$G$ construction to the Teissier distribution. The proposed distribution provides additional shape flexibility for modelling positive-valued observations and accommodates different density and hazard-rate shapes. Several statistical properties, including moments, the moment generating function, quantile function, and order statistics, are derived. Maximum likelihood, maximum product spacing, and Bayesian methods are considered for parameter estimation, and simulation results demonstrate improving estimation accuracy with increasing sample size. A median-based regression model is further developed for the KTD, with its finite-sample performance examined through simulation. The practical utility of the proposed regression model is demonstrated through an application to monthly gridded rainfall data from the Northwest Himalayas, incorporating annual and semi-annual harmonic terms to capture seasonal variation. Its performance is assessed against Kumaraswamy and beta regression using residual-based goodness-of-fit diagnostics and likelihood-based model selection. The results demonstrate effectiveness of the proposed distribution for modelling heterogeneous and skewed positive-valued rainfall data.
	\end{abstract}


\section{Introduction} \label{intro}
Statistical modeling of positive-valued observations exhibiting pronounced asymmetry, heterogeneity, and non-normality is important in environmental, hydrological, and reliability studies. For data defined on the unit interval \((0,1)\), the beta distribution is widely used because of its flexibility. However, its cumulative distribution function (CDF) involves the incomplete beta function, which can make analytical calculations and simulation less tractable. To address this limitation, Kumaraswamy \cite{kumaraswamy1980generalized} introduced a two-parameter distribution for double-bounded random processes, particularly with applications in hydrology. The Kumaraswamy (Kw) distribution has density and distribution functions
\[f(x)=abx^{a-1}(1-x^{a})^{b-1}\quad \text{and}\quad F(x)=1-(1-x^{a})^{b}\]
where \(a,b>0\) and \(0<x<1\).

Jones \cite{jones2009kumaraswamy} examined the properties of the Kw distribution and highlighted its similarity to the beta distribution together with its greater analytical tractability. In hydrology and related fields, the Kw distribution has been widely applied \cite{ganji2006grain, fletcher1996estimation, sundar1989application}. Nadarajah \cite{nadarajah2008distribution} noted its extensive use in hydrological studies and its role as a tractable alternative to the beta distribution; see also Koutsoyiannis and Xanthopoulos \cite{koutsoyiannis1989parametric}. The flexibility of a probability model can be further enhanced by introducing additional parameters through suitable distribution-generating constructions. Such extensions have been widely explored for modelling complex and skewed data \cite{cordeiro2012gumbel, cordeiro2012kumaraswamy, ogunde2025half, dhungana2025generalized}. For positive-valued observations, the Teissier distribution \cite{teissier1934recherches} provides a suitable baseline with CDF given by:
	\begin{eqnarray}\label{td}
	G(\x) = 1-e^{\fx}; \qquad \th>0,  \x>0.
\end{eqnarray}
where, $\fx =  \th \x - e^{\th \x} +1 $. 

Several extensions of the Teissier distribution have been proposed to increase its flexibility for modelling asymmetric environmental data, including the Exponentiated Teissier distribution (ETD) by Sharma et al. \cite{sharma2022exponentiated}, followed by the Lehmann type-II Teissier distribution \cite{kumaran2023lehmann} and the Alpha Power ETD \cite{poonia2022alpha}, the latter being used for positively skewed, leptokurtic, and heavy-tailed annual maximum rainfall. Teissier-based models have also been considered for flood levels \cite{krishnan2022assessment} and seasonal rainfall through univariate and bivariate extensions \cite{thakur2022uni}. To further advance these developments and accommodate positive-valued data, the Kw-(G) framework can be combined with the Teissier distribution as the baseline, leading to the proposed three-parameter Kumaraswamy--Teissier distribution (KTD). This construction is distinct from the recently introduced Odd Teissier Kumaraswamy (OTKw) distribution of Nzei et al. \cite{nzei2026odd}, which uses the Teissier distribution as the generator and the Kumaraswamy distribution as the baseline within an odd-ratio T-X framework and restricted to \((0,1)\). 

Beyond distributional flexibility, incorporating covariate information into distributional parameters provides a natural framework for modelling systematic variation in data. Several parametric regression models have been developed in this direction; for example, Ferrari and Cribari-Neto \cite{ferrari2004beta} introduced beta regression by linking the mean to explanatory variables, while Mitnik and Baek \cite{mitnik2013kumaraswamy} developed a median-dispersion parameterization of the Kumaraswamy distribution for regression modelling. Related developments include \cite{bhati2024new, el2025unit, korkmaz2021unit}. For seasonal data, distributional regression can further incorporate systematic variation through time-dependent covariates. Perreault et al. \cite{perreault2025distributional} proposed a seasonal distributional regression framework in which distributional characteristics vary with the time of year through Fourier terms, without explicitly specifying the temporal dependence of the underlying process. The need for such a flexible regression framework is particularly relevant to rainfall over the Northwest Himalayas (NWH), where complex topography and atmospheric processes lead to strong spatial and temporal variability \cite{sharma2023hydrological, pawar2015role}. Rainfall also exhibits pronounced skewness, heavy tails, and localized extremes \cite{gouda2023extreme, jena2021novel, sana2024preliminary}. These characteristics motivate a flexible distributional model that can accommodate positive, non-normal rainfall while accounting for seasonal variation through covariates \cite{mishra2026spatiotemporal}. Motivated by these developments, we formulate a median-based regression model for the KTD, allowing its conditional median to vary with explanatory variables while retaining the flexibility of the proposed distribution. 

The main contributions of this work are threefold. First, we construct the Kumaraswamy--Teissier distribution by applying the Kumaraswamy-\(G\) framework to the Teissier distribution in Section~\ref{model} and derive its principal statistical properties in Section~\ref{properties}. The maximum likelihood, maximum product spacing, and Bayesian estimation procedures are presented in Section~\ref{estimation}, with their finite-sample performance investigated through simulation in Section~\ref{simulation}. Second, we develop a median-based regression model for the KTD in Section~\ref{reg}, together with its likelihood inference and diagnostic framework, and evaluate the finite-sample performance of the regression estimators through simulation. Third, the practical applicability of the proposed regression model is demonstrated in Section~\ref{data} through the analysis of monthly gridded rainfall over the Northwest Himalayas, where its performance is compared with Kumaraswamy and beta regression under a common seasonal regression structure. Section~\ref{conclusion} summarizes the main findings and conclusions of the study.

\section{Kumaraswamy Teissier distribution (KTD)} \label{model}
Let $G(\x)$ be an arbitrary baseline CDF, then the CDF of the Kw-G distribution is given by
	\begin{eqnarray}\label{kum}
		\F(\x) = 1-\left\{ 1-G(\x)^\a \right\}^\b;  \qquad \a,\b>0.
	\end{eqnarray}
By using Teissier as the baseline CDF given in equation (\ref{td}), we get the CDF of KTD as below
	\begin{eqnarray}\label{ktcdf}
		\F(\x) = 1-\left( 1-\left( 1-e^{\fx} \right)^\a \right)^\b, \qquad \x>0
	\end{eqnarray}
	where, $\a,\b,\th>0$. And the corresponding PDF is
	\begin{eqnarray}\label{ktpdf}
		\f(\x) = \a \b \th (e^{\th \x} -1) e^{\fx} \left(1-e^{\fx} \right)^{\a-1} \left(  1-\left(1-e^{\fx} \right)^{\a} \right)^{\b-1}.
	\end{eqnarray}
	Figure \ref{Figure 10} illustrates both the PDF and hazard rate function (HRF) plots for the KTD.\\
	Following are the sub-models of KTD
	\begin{enumerate}
		\item If $\b=1$ in Equation (\ref{ktpdf}) we get
		\begin{eqnarray*}
			\f(\x) = \a \th (e^{\th \x} -1) e^{\fx} \left(1-e^{\fx} \right)^{\a-1},
		\end{eqnarray*}
		which is PDF of ETD with $\a>0$ shape and $\th>0$ scale parameters.
		
		\item If $\a=\b=1$, Equation (\ref{ktpdf}) reduced to the PDF of TD as
		\begin{eqnarray*}
			\f(\x) = \th (e^{\th \x} -1) e^{\fx}.
		\end{eqnarray*}
	\end{enumerate}
	\begin{corollary} Suppose $\X \sim KTD(\pv)$. Then, the survival function $S(\x)$, HRF $(h(\x))$ and reverse HRF $(r(\x))$ of $\X$ are given by
		\begin{eqnarray*} \label{survival}
			S(\x) = \left(1-\left(1-e^{\fx}\right)^a \right)^b,
		\end{eqnarray*}
		\begin{eqnarray*} \label{hrf}
			h(\x) = \frac{ab \th (e^{\th \x}-1) e^{\fx} \left(1-e^{\fx}\right)^{a-1}}{1-\left(1-e^{\fx}\right)^a},
		\end{eqnarray*}
		and
		\begin{eqnarray*} \label{reverse hrf}
			r(\x) = \frac{ab \th (e^{\th \x}-1)e^{\fx} \left(1-e^{\fx}\right)^{a-1} \left(1-\left(1-e^{\fx}\right)^{a}\right)^{b-1}}{1-\left(1-\left(1-e^{\fx}\right)^a\right)^b}
		\end{eqnarray*}
		respectively. 
	\end{corollary}

	\begin{figure}[!ht]
		\centering
		\subfloat[]{\includegraphics[width=0.49\textwidth]{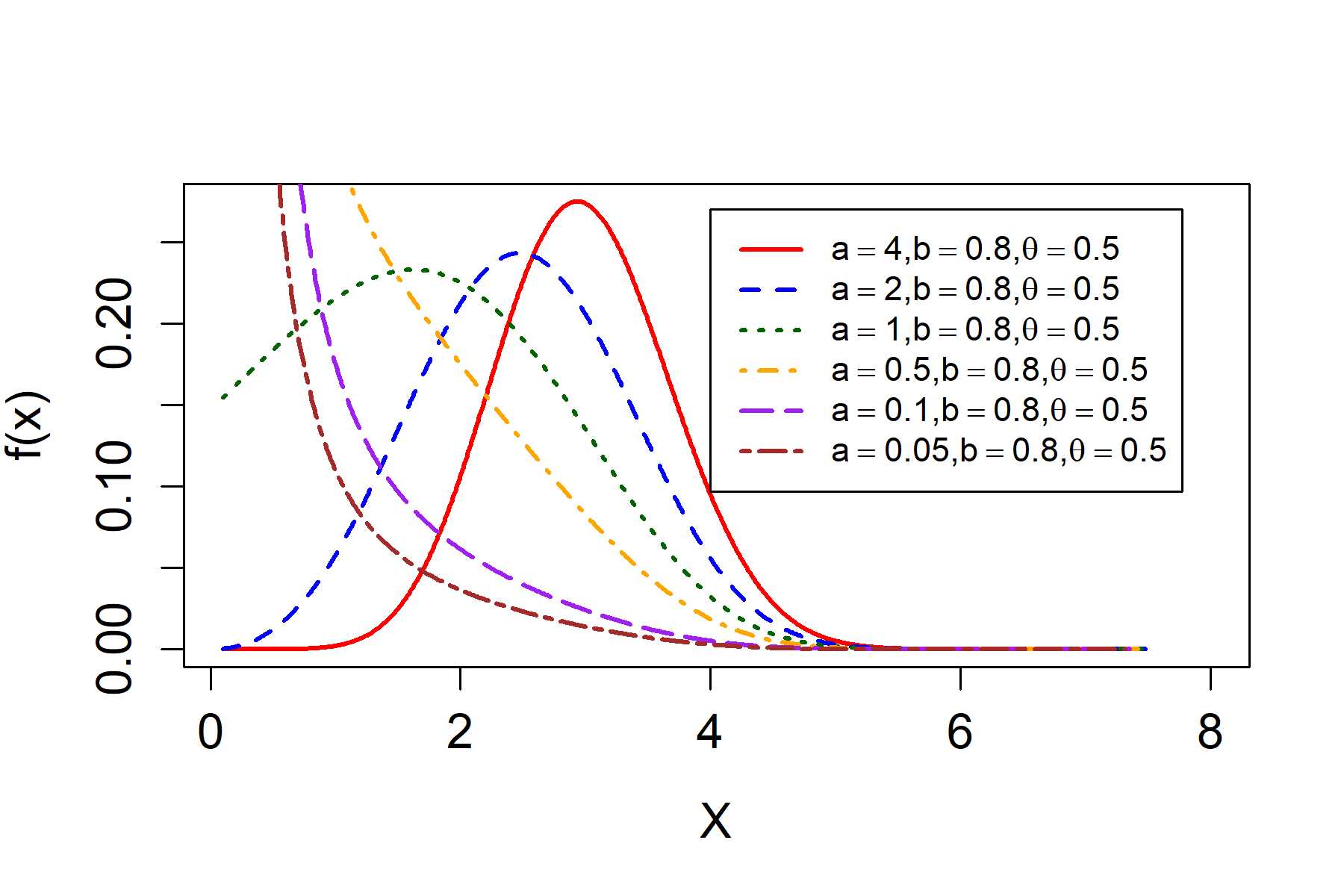}\label{fig1:f1}}
		\hfill
		\subfloat[]{\includegraphics[width=0.49\textwidth]{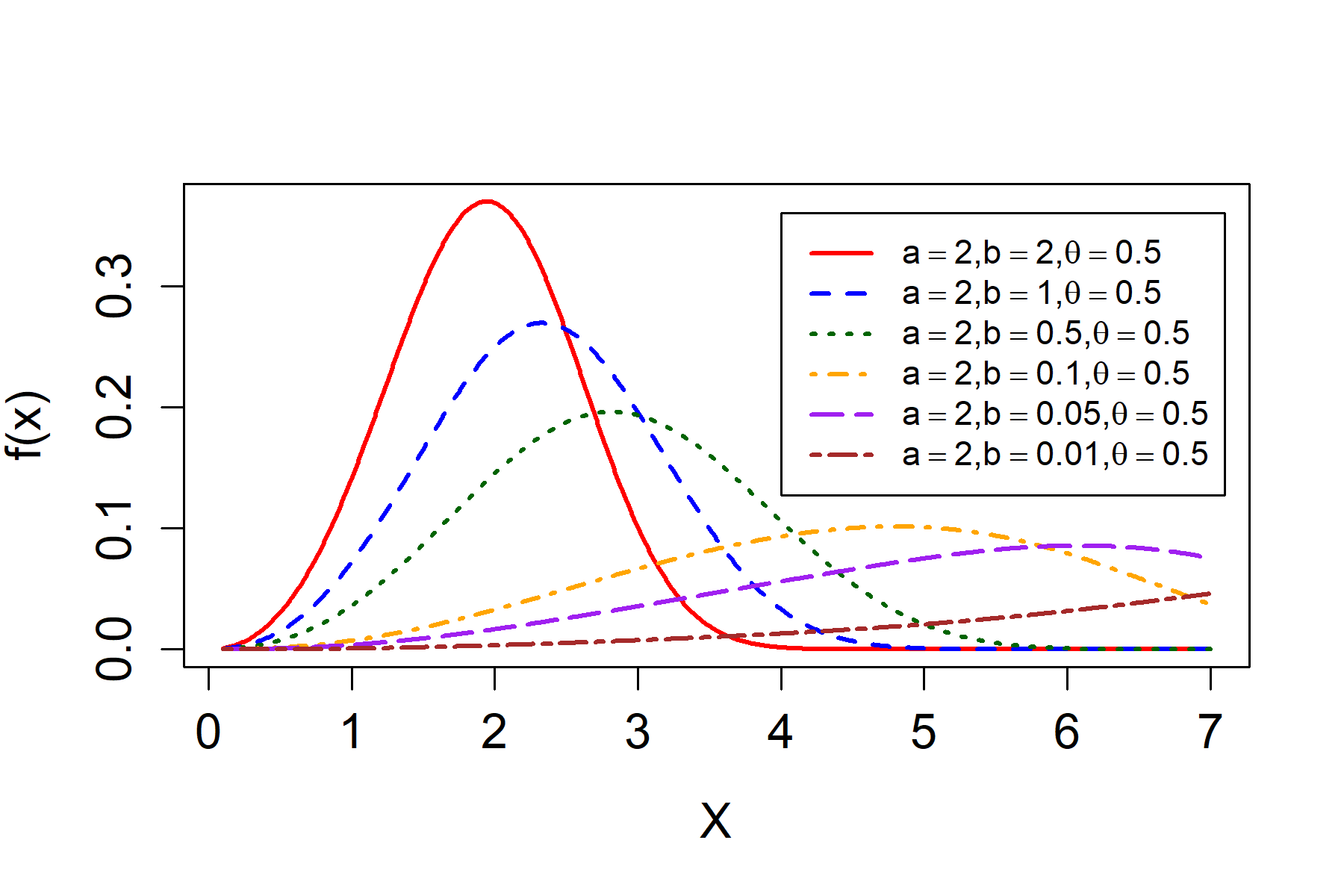}\label{fig1:f2}}
		\hfill
		\subfloat[]{\includegraphics[width=0.49\textwidth]{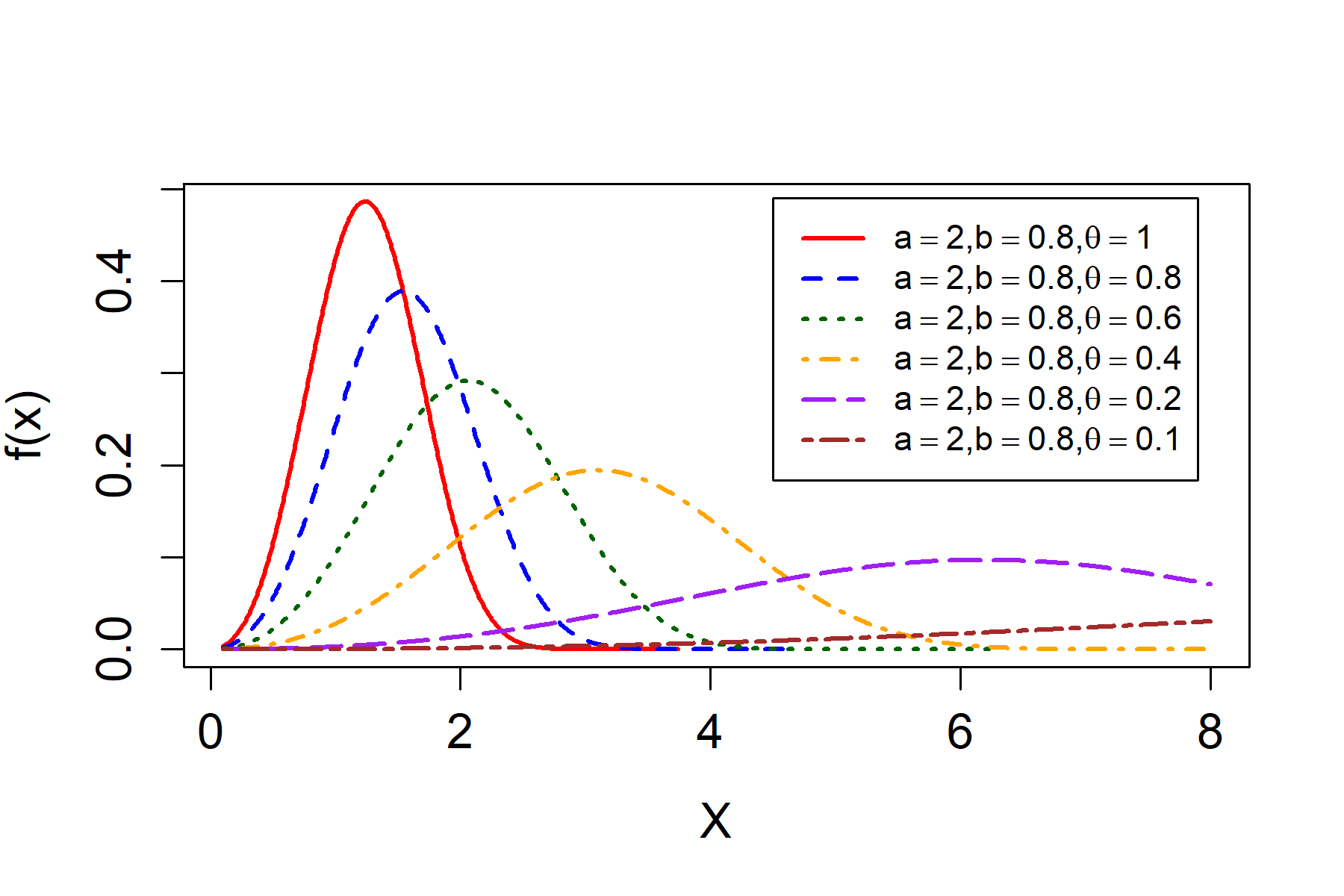}\label{fig1:f3}}
		\hfill
		\subfloat[]{\includegraphics[width=0.49\textwidth]{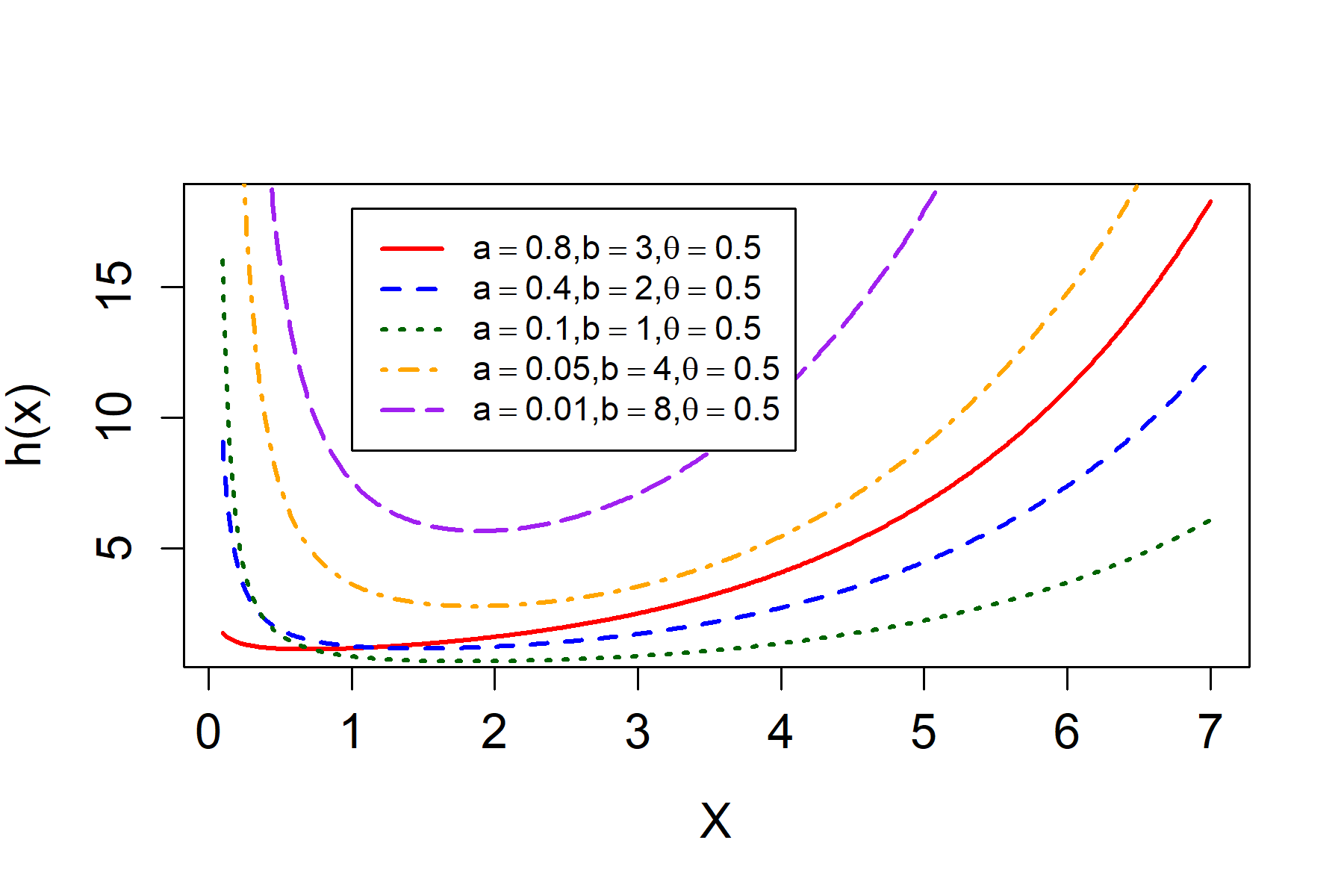}\label{fig1:f4}}
		\hfill
		\subfloat[]{\includegraphics[width=0.49\textwidth]{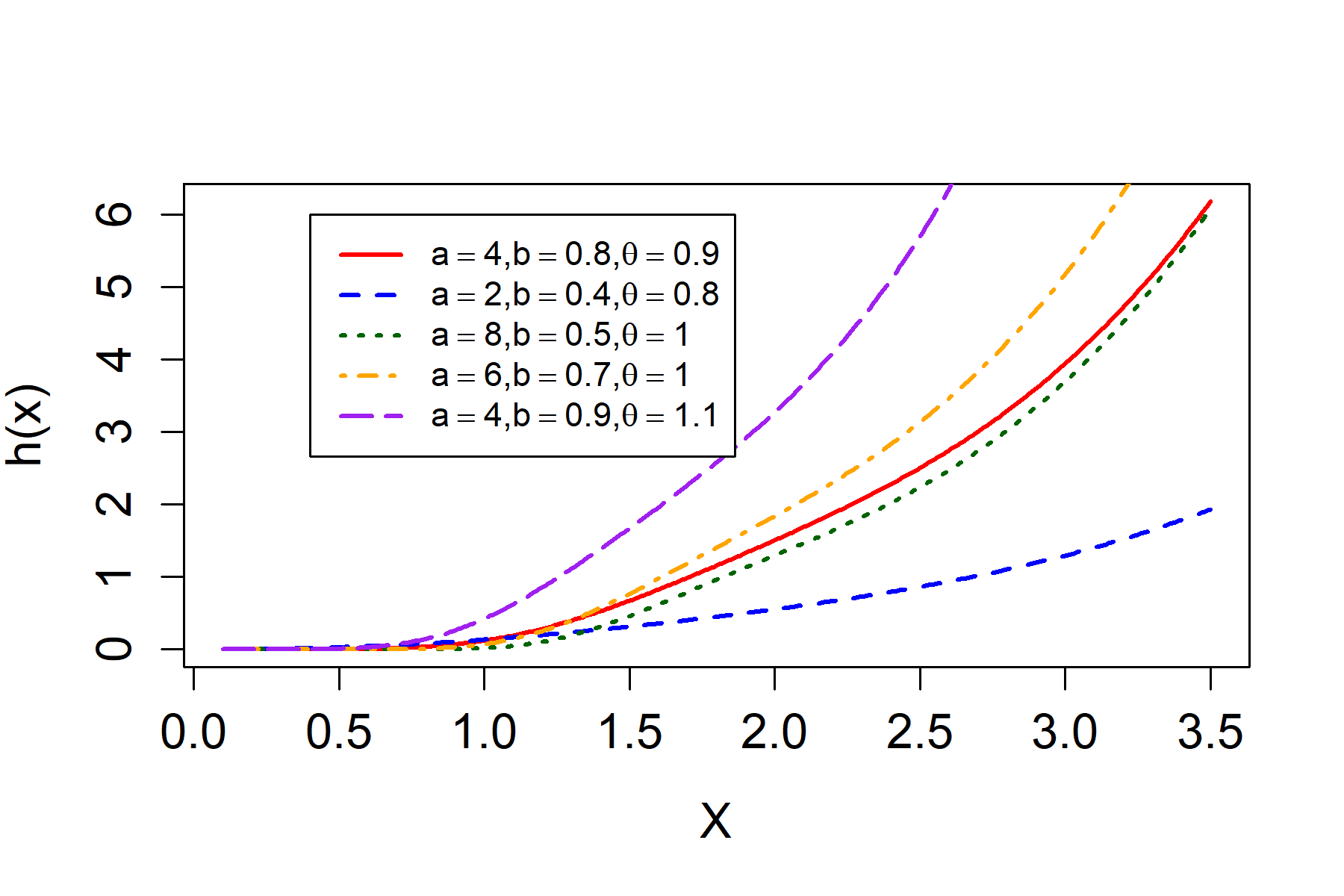}\label{fig1:f5}}
		\hfill
		\subfloat[]{\includegraphics[width=0.49\textwidth]{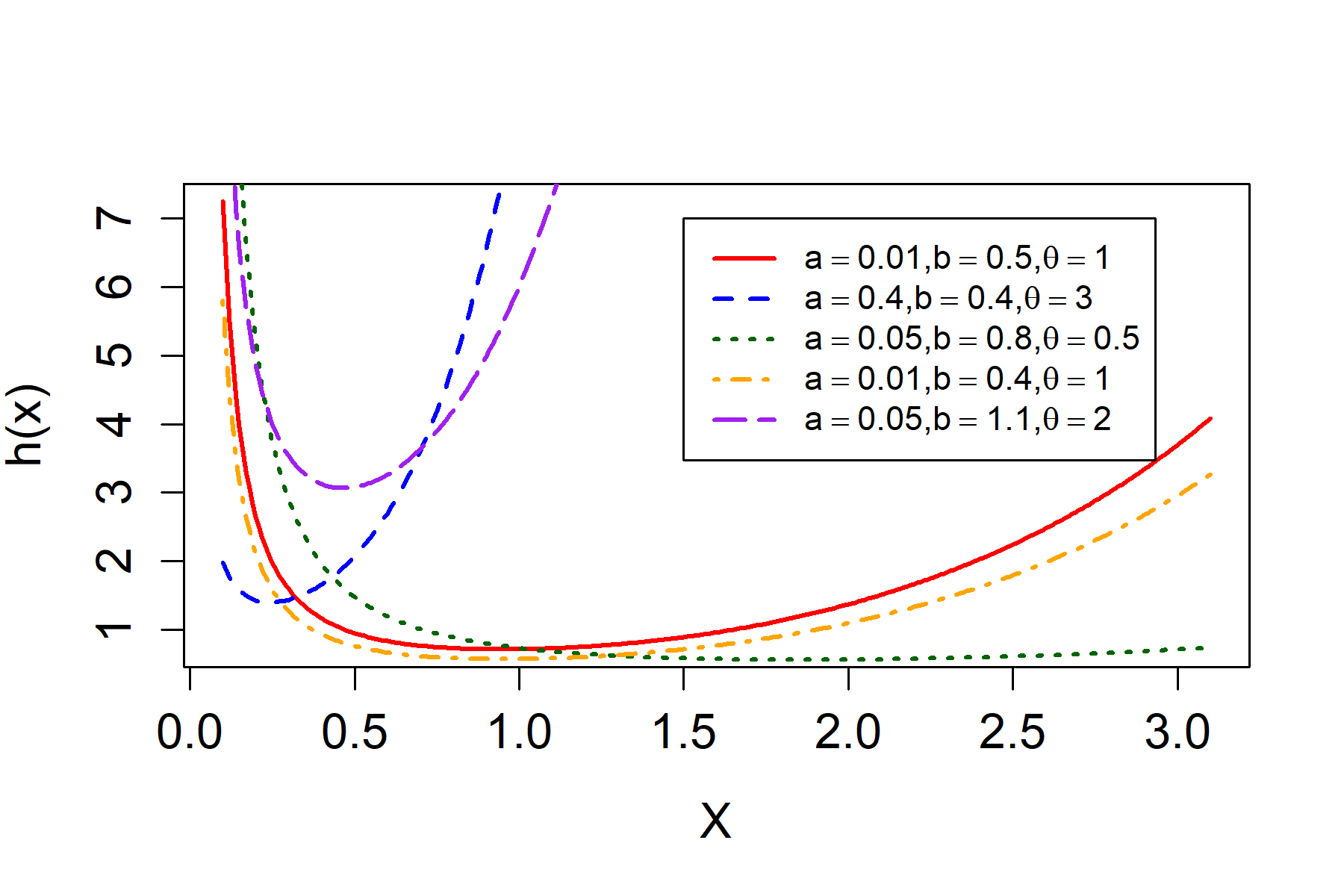}\label{fig1:f6}}
		\hfill
		\caption{Figures (a), (b), and (c) exhibits the density plots of the proposed model at different parameter values; Figure (d) exhibits the hazard rates of the proposed model for \a <1 and \b >1; Figure (e) exhibits hazard rate for \a >1 and \b <1; Figure (f) exhibits hazard rate of the model at varying values of the parameters.}
		\label{Figure 10}
	\end{figure}
	
	\subsection{Expansion for the CDF}
	For any \b > 0, we can give series expansion of (\ref{ktcdf}) by using binomial theorem:
	\begin{eqnarray} \label{exapansioncdf}
		\F(\x) &= 1- \sum_{j=0}^{\infty} t_j \left(1-e^{\fx} \right)^{\a j}
	\end{eqnarray}
	where $t_j = (-1)^j \binom{b}{j}$ and $\sum_{j=0}^{\infty} t_j = 1$. If $\b \in \mathbb{Z}$, then the sum stops at b. 
	
	\subsection{Expansion for the PDF}
	\begin{align*}
		\f(\x) &= \a \b \th (e^{\th \x} -1) e^{\fx} \left(1-e^{\fx} \right)^{\a-1} \left(  1-\left(1-e^{\fx} \right)^{\a} \right)^{\b-1}\\
		\text{Taking}~~~ g(\x)= \frac{d}{dx}G(\x),\\
		\f(\x)&= \a \b g(\x) G(\x)^{\a -1} \sum_{j=0}^{\infty} (-1)^j \binom{b-1}{j}G(\x)^{\a j}\\
		&= g(\x) \sum_{j=0}^{\infty} t_j G(\x)^{\a (j+1)-1}\\
	\end{align*}
	where $t_j = \a \b (-1)^j \binom{b-1}{j}$ and $\sum_{j=0}^{\infty} t_j = 1$. If $\b \in \mathbb{Z}$, then the sum stops at \b -1. If a is a real non-integer then, expansion can be further expanded as:
	\begin{align*}
		\f(\x) &= g(\x) \sum_{i,j=0}^{\infty}\sum_{k=0}^{i} \eta_{i,j,k} G(\x)^{k}
	\end{align*}
	where $\eta_{i,j,k} = \a \b (-1)^{i+j+k} \binom{b-1}{j} \binom{\a (j+1)-1}{i} \binom{i}{k}$ and $\sum_{i,j=0}^{\infty}\sum_{k=0}^{i} \eta_{i,j,k} = 1$. The probability density function can be expressed as (\ref{exapnsionpdf}) given below:
	\begin{eqnarray} \label{exapnsionpdf}
		\f(\x) &= \th (e^{\th \x} -1) e^{\fx} \sum_{i,j=0}^{\infty}\sum_{k=0}^{i} \eta_{i,j,k} \left(1-e^{\fx} \right)^{k}
	\end{eqnarray}

	\section{Statistical Properties} \label{properties}
	\begin{theorem}
		Let $\X \sim KTD(\pv)$. Then the $r^{th}$ order moment is given by 
		\begin{eqnarray*}
			E(X^r) = \sum_{i,j=0}^\infty \sum_{k=0}^i \sum_{l=0}^k \eta_{i,j,k} \binom{k}{l} (-1)^l e^{l+1} \frac{1}{\theta^r} \left[ J(l+2, l+1) - J(l+1, l+1) \right].  
		\end{eqnarray*}
	\end{theorem}
	\noindent The demonstration of the \(r^{th}\) order moment is provided in \cref{appendixB}, where we employ the definition of moments and then integrate the expanded form of the PDF.
	
	\begin{theorem}
		Let $\X \sim KTD(\pv)$. Then the mgf of $\X$ is given by
		\begin{eqnarray*}
			M_X(t) &= \sum_{i,j=0}^\infty \sum_{k=0}^i \sum_{l=0}^k \eta_{i,j,k} \binom{k}{l} (-1)^l e^{l+1} \Bigg[ (l+1)^{-(\frac{t}{\theta} + l + 2)} \Gamma\!\left( \frac{t}{\theta} + l + 2, \, l+1 \right) \nonumber \\
			&\qquad - (l+1)^{-(\frac{t}{\theta} + l + 1)} \Gamma\!\left( \frac{t}{\theta} + l + 1, \, l+1 \right) \Bigg].
		\end{eqnarray*}
	\end{theorem}
	\noindent The proof of mgf is presented in \cref{appendixC}, where we utilize the definition of mgf and subsequently integrate the expanded form of the PDF.

	\subsection{Quantile Function}
	The $p^{th}$ quantile of proposing model can be given by:
	\begin{eqnarray} \label{quantile}
		x_p = \frac{1}{\th}~ ln \left( -W_{-1} \left( \frac{ \left(1-(1-p)^{1/b} \right)^{1/a}-1}{e} \right) \right)
	\end{eqnarray}
	The median M can be obtained by taking p = 0.5
	\subsection{Order Statistics}
	The density function $f_{m:n}(x)$ of the $m^{th}$ order statistic, where $m = 1, \ldots, n$ derived from independent and identically distributed (i.i.d.) random variables $X_1, \ldots, X_n$ following KTD, as computed in \cref{appendixD}, can be expressed as:
	\begin{align*}
		f_{m:n}(x)&= \th (e^{\th \x} -1) e^{\fx} \sum_{i,j,r=0}^{\infty}\sum_{k=0}^{i} \tau_{m:n} \left(1-e^{\fx} \right)^{k}
	\end{align*}
	where $$\tau_{i,j,k,r}^{(m:n)} = \frac{n!}{(m-1)!(n-m)!} \a \b (-1)^{i+j+k+r}  \binom{m-1}{r} \binom{b(r+n-m+1)-1}{j} \binom{\a (j+1)-1}{i} \binom{i}{k} $$
	\textbf{Case 1:} When $m=1 \Rightarrow r=0$. So, the  PDF of the smallest order statistic $X_{(1)}$ can be given by:
	$$f_{1:n}(x) = n \a \b \th (e^{\th \x} -1) e^{\fx}  \sum_{i,j=0}^{\infty}\sum_{k=0}^{i} (-1)^{i+j+k}  \binom{bn-1}{j} \binom{\a (j+1)-1}{i} \binom{i}{k} \left(1-e^{\fx} \right)^{k} $$
	\textbf{Case 2:} When $m=n$, the  PDF of the largest order statistic $X_{(n)}$ can be given by:
	$$f_{n:n}(x) = n \a \b \th (e^{\th \x} -1) e^{\fx} \sum_{i,j,r=0}^{\infty}\sum_{k=0}^{i} (-1)^{i+j+k+r} \binom{n-1}{r}  \binom{b(r+1)-1}{j} \binom{\a (j+1)-1}{i} \binom{i}{k} \left(1-e^{\fx} \right)^{k}$$
	
	\begin{corollary}
		Let $\X_k \sim KTD(\th,\a,\b_k)$ for $1\leq k$ be independent random variables. Then $X_{(1)}\sim KTD(\th,\a,\sum_{k=0}^{n} \b_k)$.
	\end{corollary}

	\section{Parameter Estimation} \label{estimation}
	In this section, we discuss the Maximum likelihood (ML) estimation, Maximum product spacing (MPS) and Bayesian analysis for the parameter estimation of the KTD distribution. 
	\subsection{Maximum Likelihood Estimation} 
	Let $x_1, \ldots, x_n$ be random sample of size $n$ from KTD($\pv$). The log-likelihood function for the parameter $\pv$ is given by
	\begin{eqnarray} \label{MLE}
		\mathcal{L}(\pv) &=& \ss \log(\a \b \th) + \sum_{\i=1}^{\ss} \log(e^{\th \x_i} -1) +  \sum_{\i=1}^{\ss} \fxi + (\a-1)\sum_{\i=1}^{\ss}\log y_i \nonumber \\ &&  + (\b-1)\sum_{\i=1}^{\ss} \log\left(1- y_i ^\a \right) 
	\end{eqnarray}
	where, $ 1-e^{\fxi} = y_i$
	\begin{align*}
		\mathcal{L}_a(\pv)  = \frac{\partial \mathcal{L}(\pv)}{\partial \a} &= \frac{\ss}{\a} + \sum_{\i=1}^{\ss} \log y_i \left[1- \frac{ (\b-1)}{y_i^{-\a}-1} \right] = 0 \\
		\mathcal{L}_b(\pv)  =\frac{\partial \mathcal{L}(\pv)}{\partial \b} &= \frac{\ss}{\b} + \sum_{\i=1}^{\ss} \log\left(1- y_i ^\a \right) = 0 \\
		\mathcal{L}_\th(\pv)  =\frac{\partial \mathcal{L}(\pv)}{\partial \th} &= \frac{\ss}{\th} + \sum_{\i=1}^{\ss} \frac{e^{\th x_i} x_i}{\left(e^{\th x_i}-1 \right)} - \sum_{\i=1}^{\ss} x_i \left(e^{\th x_i}-1 \right)+ \sum_{\i=1}^{\ss} \frac{x_i e^{\fxi} (e^{\th \x_i} -1)}{y_i} \left[ \a-1 - \frac{\a (\b-1)}{y_i^{-\a}-1} \right] = 0
	\end{align*}
	The solutions to the nonlinear equations $\nabla \mathcal{L}(\Theta) = 0$ yield the ML estimates of the parameters. These equations are solved iteratively. The observed information matrix can be given by 
	\[ J(\Theta)
	=
	\left[ {\begin{array}{ccc}
			\mathcal{L}_{aa} & \mathcal{L}_{ab} & \mathcal{L}_{a\th} \\
			. & \mathcal{L}_{bb} & \mathcal{L}_{b \th}\\
			. & . & \mathcal{L}_{\th \th}\\
	\end{array} } \right]
	\]
	and its elements are
	\begin{align*}
		\mathcal{L}_{aa}(\pv) &= \frac{-n}{a^2} - (b-1) \sum_{\i=1}^{\ss}\log^2 (y_i)  \left[\frac{y_i^a}{(1-y_i^a)^2} \right] \\
		\mathcal{L}_{bb}(\pv) &= \frac{-n}{b^2}\\
		\mathcal{L}_{\th \th}(\pv) &= \frac{-n}{\th ^2} - \sum_{i=1}^n x_i^2 e^{\th x_i} \biggl\{1 + \frac{1}{(e^{\th x_i} -1)^2 }\biggl\} - a^2 (b-1)\sum_{i=1}^n \frac{y_i^a}{(1-y_i^a)^2} \biggl\{ \frac{x_i e^{\fxi} (e^{\th \x_i} -1)}{y_i}  \biggl\}^2 \\ &+ \sum_{i=1}^n \left[ \a-1 - \frac{\a (\b-1)}{y_i^{-\a}-1} \right] \biggl\{  x_i^2 e^{\th x_i} \left( -1 + \frac{1}{y_i} \right) -\frac{x_i^2 (e^{\th x_i}-1)^2 e^{\fxi}}{y_i^2}
		\biggl\}\\
		\mathcal{L}_{ab}(\pv) &= \sum_{\i=1}^{\ss} \log y_i \left[\frac{1}{1-y_i^{-\a}} \right]\\
		\mathcal{L}_{a\th}(\pv) &= \sum_{\i=1}^{\ss} \frac{x_i e^{\fxi} (e^{\th \x_i} -1)}{y_i} \left[ 1 - \frac{ (\b-1)}{y_i^{-\a}-1} \biggl\{1 + \frac{a \log y_i}{1-y_i^a} \biggl\} \right]\\
		\mathcal{L}_{b\th}(\pv) &= -\a~\sum_{\i=1}^{\ss} \frac{x_i e^{\fxi} (e^{\th \x_i} -1)}{y_i} \left[\frac{1}{y_i^{-\a}-1} \right]
	\end{align*}

	\subsection{Maximum Product Spacing}
	The estimates can be derived by maximizing the geometric mean of the uniform spacing, expressed as:
	$$G = \left[ \prod_{i=1}^{\ss +1} D_i \right]^{1/(n+1)}$$
	where, $$D_i = F(x_{(i)}) - F(x_{(i-1)})$$
	Equivalently by maximizing the following function
	\begin{eqnarray}
		\mathcal{P}(\pv) = \frac{1}{n+1}  \sum_{\i=1}^{\ss +1} \log D_i
	\end{eqnarray}
	\begin{align*}
		\frac{\partial \mathcal{P}(\pv)}{\partial \a} &= \frac{1}{n+1}  \sum_{\i=1}^{\ss +1} \frac{\xi_\a({x_{(i)}}) - \xi_\a({x_{(i-1)}})}{D_i} = 0\\
		\frac{\partial \mathcal{P}(\pv)}{\partial \b} &= \frac{1}{n+1}  \sum_{\i=1}^{\ss +1} \frac{\xi_\b({x_{(i)}}) - \xi_\b({x_{(i-1)}})}{D_i} = 0\\
		\frac{\partial \mathcal{P}(\pv)}{\partial \th} &= \frac{1}{n+1}  \sum_{\i=1}^{\ss +1} \frac{\xi_\th({x_{(i)}}) - \xi_\th({x_{(i-1)}})}{D_i} = 0
	\end{align*}
	where,
	\begin{align*}
		\xi_a({x_{(i)}})  &= b~ y_{(i)}^a~ \log y_{(i)} \left( 1- y_{(i)}^a \right)^{b-1}\\
		\xi_b({x_{(i)}})  &= -b \left( 1- y_{(i)}^a \right)^{b-1}\\
		\xi_\th({x_{(i)}}) &= \a\b x_{(i)}~ y_{(i)}^{a-1} ~ e^{\phi(\x_{(i)}; \th)} \left( e^{\th x_{(i)} } -1\right) \left( 1- y_{(i)}^a \right)^{b-1}
	\end{align*}

	\subsection{A Bayesian Analysis}
	$$\pi (\a,\b,\th) = \pi(\a) \times \pi(\b) \times \pi(\th) $$
	Here, $\a \sim \Gamma (p_1, q_1), ~\b \sim \Gamma (p_2, q_2),~ \th \sim \Gamma (p_3, q_3)$ and $\Gamma (p_i, q_i)$ denotes gamma distribution with mean $p_i/{q_i}$ and variance $p_i/{q_i^2}$ for $p_i, q_i >0$ and density function is:
	$$f(x) = \frac{q_i^{p_i}~ \x^{~p_i-1}~ e^{-q_i~ \x}}{\Gamma (p_i)}$$
	The joint posterior pdf is given by
	\begin{align} \label{bayes}
		\pi (\a,\b,\th| \x) &\propto (\a \b \th)^n~ \prod_{i=1}^n (e^{\th x_i} -1)~ e^{\fxi} ~y_i^{a-1}~ (1-y_i^a)^{b-1} \times \pi(\a,\b,\th) 
	\end{align}
	The conditional posteriors are given by
	\begin{align*}
		\pi (\a| \x,\b,\th) &\propto \a^n~ \prod_{i=1}^n ~y_i^{a}~ (1-y_i^a)^{b-1} \times \pi(\a)\\
		\pi (\b| \x, \a, \th) &\propto \b^n~ \prod_{i=1}^n (1-y_i^a)^{b-1} \times \pi(\b)   \\
		\pi (\th| \x,\a,\b) &\propto \th^n~ \prod_{i=1}^n (e^{\th x_i} -1)~ e^{\fxi} ~y_i^{a-1}~ (1-y_i^a)^{b-1} \times \pi(\th)   
	\end{align*}
	
	\subsubsection*{Metropolis-Hastings (MH) algorithm} 
	The MH algorithm, introduced by Metropolis and Hastings, is one of the most widely used methods for obtaining Bayesian estimates. In addition, credible intervals and highest posterior density (HPD) intervals for $\a,~ \b~ \& ~\th$ are constructed using this approach. The MH algorithm is particularly effective when the posterior distribution cannot be derived analytically, as it enables the generation of posterior samples using arbitrary proposal distributions. The required samples from the posterior distribution, as specified in equation \eqref{bayes}, can be generated by following the steps detailed below. 
	\begin{enumerate}
		\item Select an initial guess for $(\a, \b, \theta)$ as $\Theta_0 = (\a_0, \b_0, \theta_0)$.
		\item Since the conditional posterior PDFs are not in closed form, generate them using a normal distribution. Specifically, draw $\a'$ from $N(\a_{n-1}, \sigma_1^2)$, $\b'$ from $N(\b_{n-1}, \sigma_2^2)$, and $\theta'$ from $N(\theta_{n-1}, \sigma_3^2)$, where $\sigma_1^2$, $\sigma_2^2$, and $\sigma_3^2$ represent the variances of $\a$, $\b$, and $\theta$, respectively.
		\item Compute $d = \min \left\{ 1, \frac{\pi (\Theta'| x)}{\pi (\Theta_{n-1} | x)} \right\}$.
		\item Generate a random number $r$ from $U(0, 1)$.
		\item If $r \leq d$, update $\Theta_n$ to $\Theta'$. Otherwise, set $\Theta_n = \Theta_{n-1}$.
		\item Repeat steps 2 to 5 for $N$ iterations.
	\end{enumerate}
	After burn-in starting $N_0$ samples, we get estimates as:
	\[
	\Tilde{a}_{MH} = \frac{1}{N-N_0} \sum_{j = N_0 + 1}^N a_j, \quad
	\Tilde{b}_{MH} = \frac{1}{N-N_0} \sum_{j = N_0 + 1}^N b_j, \quad
	\Tilde{\th}_{MH} = \frac{1}{N-N_0} \sum_{j = N_0 + 1}^N \th_j
	\]
	The $100 (1-\alpha)\%$ credible interval is given by $(\Th^{(\alpha/2)},\Th^{(1-\alpha/2)})$. HPD is a Bayesian credible interval that identifies the narrowest region containing \( (1 - \alpha) \) of the posterior probability. For a posterior density \( \mathbb{P}(A \mid x) \), a region \( R \) is an HPD region if $\mathbb{P}\{A \in R \mid x\} = 1 - \alpha,$ and for \( A_1 \in R \) and \( A_2 \notin R \):  
	$\mathbb{P}(A_1 \mid x) \geq \mathbb{P}(A_2 \mid x).$
	
	\section{Simulation Study} \label{simulation}
	In this study, Monte Carlo simulation experiments are conducted to assess the performance of above three estimation methods for estimating the parameters of the KTD.  
	\begin{table}[!ht]
		\centering
		\small
		\caption{Mean estimates for some parameter values}
		\label{table1}
		\begin{tabular}{@{}lccccccccc@{}}
			\toprule
			& \multicolumn{3}{c}{{MLE}} & \multicolumn{3}{c}{{MPS}} & \multicolumn{3}{c}{{Bayesian}} \\ 
			\cmidrule(lr){2-4} \cmidrule(lr){5-7} \cmidrule(lr){8-10}
			{} & $\hat{\a}$ & $\hat{\b}$ & $\hat{\th}$ & $\hat{\a}$ & $\hat{\b}$ & $\hat{\th}$ & $\hat{\a}$ & $\hat{\b}$ & $\hat{\th}$ \\
			\midrule
			\multicolumn{10}{c}{{True Values:} $\a=0.8$, $\b = 2.5$, $\theta = 1.7$} \\
			\midrule
			50 & 0.8043 & 2.8552 & 1.8856 & 0.8597 & 2.7017 & 2.0339 & 0.7867 & 2.5318 & 1.8867\\
			100 & 0.7988 & 2.9588 & 1.8069 & 0.8249 & 2.8157 & 1.9091 & 0.789 & 2.5669 & 1.8538\\
			200 & 0.795 & 2.938 & 1.7723 & 0.8105 & 2.8728 & 1.8254 & 0.788 & 2.5372 & 1.7813\\
			300 & 0.7959 & 2.9254 & 1.7487 & 0.8055 & 2.8158 & 1.802 & 0.791 & 2.6616 & 1.7757\\
			\midrule
			\multicolumn{10}{c}{{True Values:} $\a=2.8$, $\b = 0.7$, $\theta = 2.3$} \\
			\midrule 
			50 & 3.1128 & 0.6552 & 2.4605 & 3.4574 & 0.7051 & 2.4526 & 3.1539 & 0.6722 & 2.4839 \\
			100 & 3.0556 & 0.6762 & 2.4259 & 2.321 & 0.6736 & 2.4595 & 3.027  & 0.6957 & 2.4411 \\
			200 & 3.0044 & 0.6768 & 2.4137 & 3.2162 & 0.6156 & 2.505 & 2.9634 & 0.69 & 2.4283\\
			300 & 2.9774 & 0.6678 & 2.4209 & 3.1422 & 0.6006 & 2.511 & 2.9503 & 0.7094 & 2.4078\\
			\midrule
			\multicolumn{10}{c}{{True Values:} $\a=2$, $\b = 1.7$, $\theta = 0.9$} \\
			\midrule 
			50 & 2.228  & 1.5915  & 1.0322 & 2.4792 & 1.5887 & 1.0615 & 2.1112 & 1.4576 & 0.9898 \\
			100 & 2.1587 & 1.6754 & 1.004 & 2.3063 & 1.682 & 1.0246 & 2.0963 & 1.4924 & 0.9864 \\
			200 & 2.1117 & 1.6547 & 0.9999 & 2.2186 & 1.5153 & 1.0478 & 2.0637 & 1.4903 & 0.9895\\
			300 &  2.0934 & 1.6419 & 0.9996 & 2.1701 & 1.5292 & 1.0349 & 2.0616 & 1.5336 & 0.9908\\
			\bottomrule
		\end{tabular}
	\end{table}
	\begin{figure}
		\centering
		{\includegraphics[width=1\textwidth, height = 7cm]{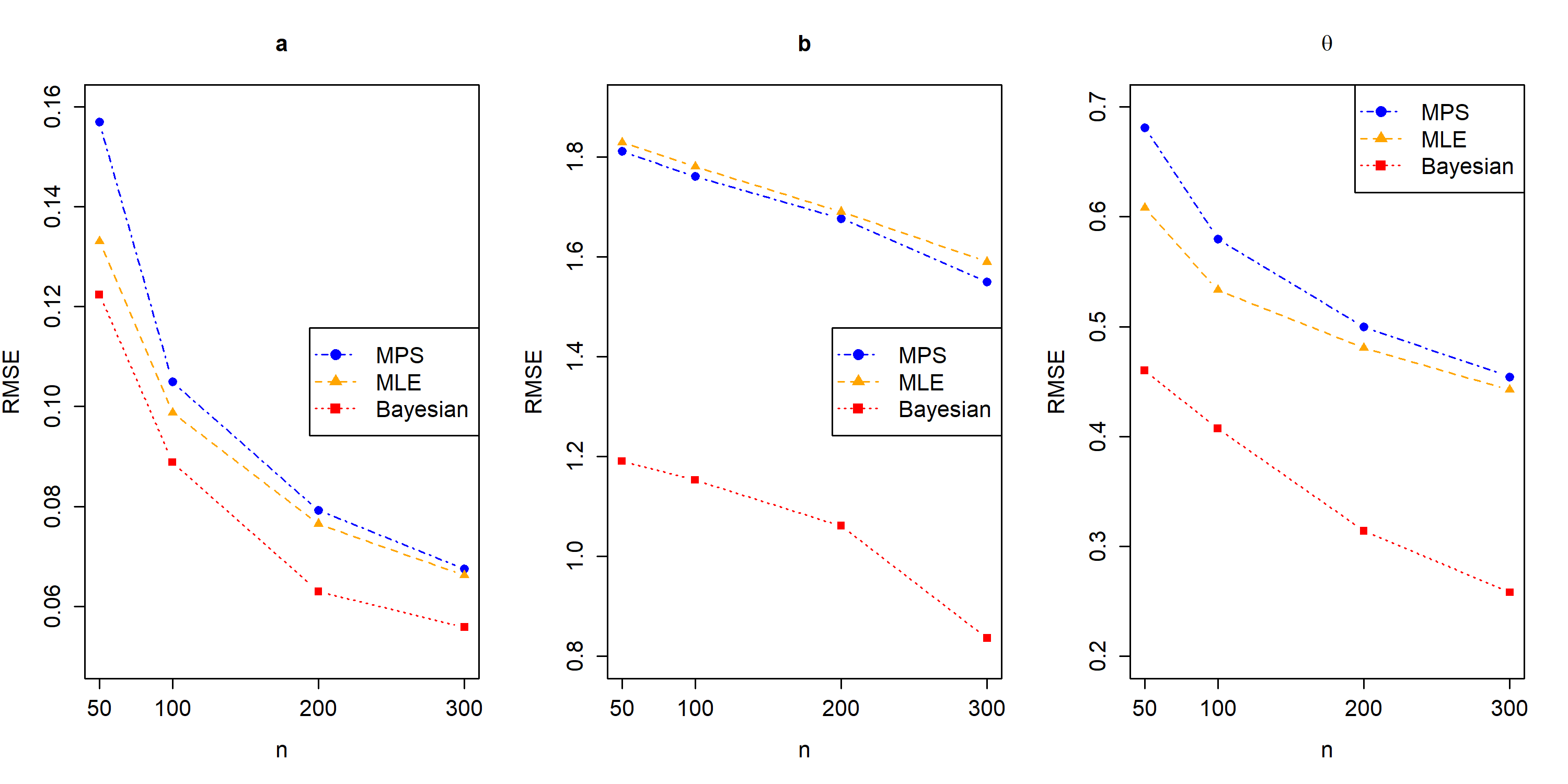}}
		\hfill
		{\includegraphics[width=1\textwidth, height = 7cm]{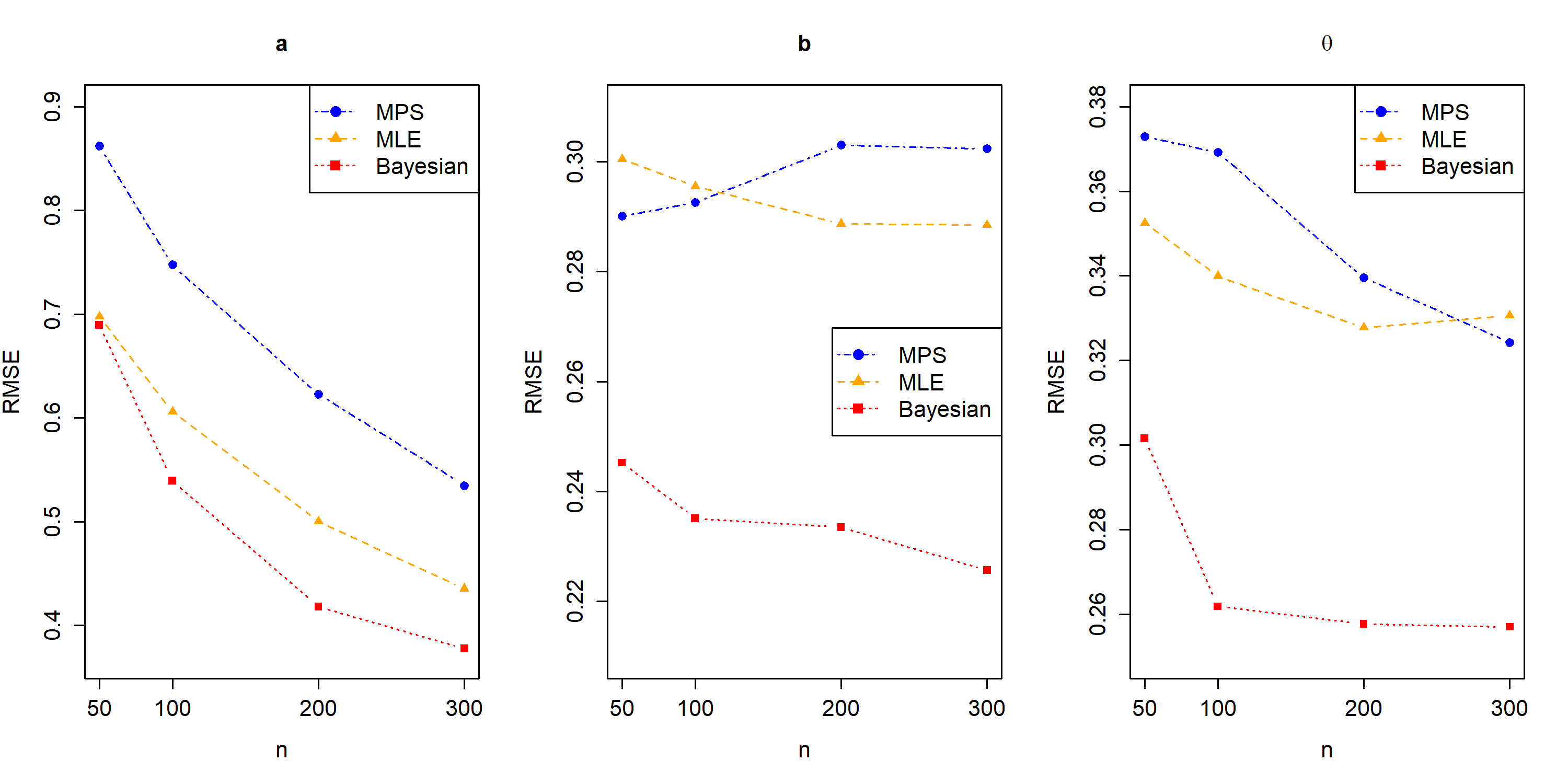}}
		\hfill
		{\includegraphics[width=1\textwidth, height = 7cm]{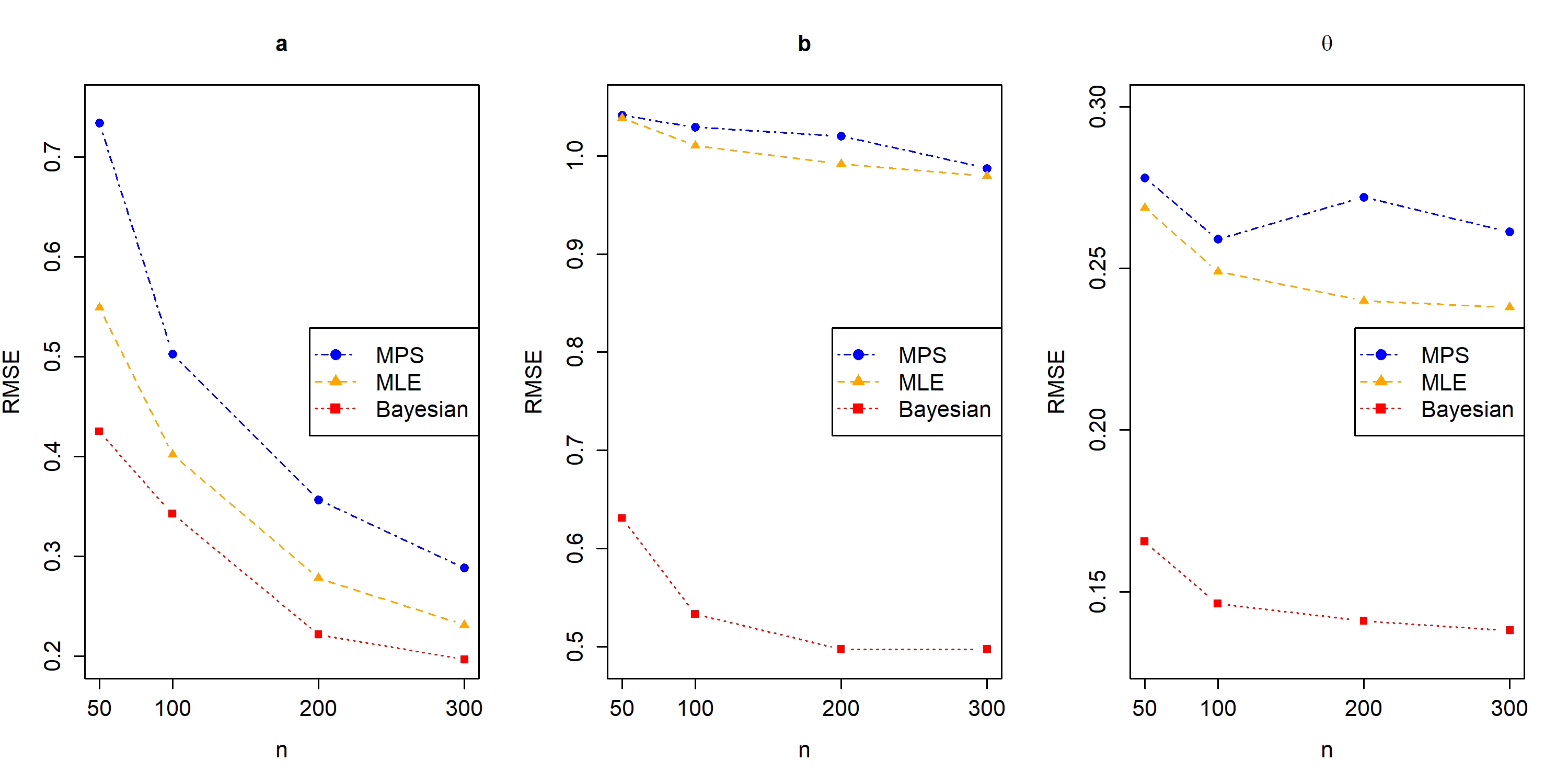}}
		\hfill
		\caption{RMSE plots}
		\label{Figure 3}
	\end{figure}
	The simulations are implemented using the R programming language, with the number of Monte Carlo replications set to \( N = 1000 \) for MLE and MPS methods. Four different sample sizes are considered: \( n = 50 \), \( 100 \), \( 200 \), and \( 300 \), along with various combinations of model parameters to ensure the robustness of the analysis. For each sample size, the evaluation of the MLE and MPS estimates are conducted using two key metrics: the mean and the root mean squared error (RMSE), calculated across the 1000 replications.

	For Bayesian estimation, the simulation involved \( N = 10,000 \) samples, with a burn-in period of \( N_0 = 1000 \) to ensure convergence before calculating posterior summaries. Bayesian estimation utilized prior distributions appropriate for the parameters of the KTD distribution, and posterior estimates are derived using Markov Chain Monte Carlo (MCMC) techniques. \\
	The results demonstrate that all three methods provide accurate and stable estimates for the KTD parameters, as indicated by the minimal RMSE values and means that closely approximate the true parameter values in Table \ref{table1} and Figure \ref{Figure 3}. Additionally, the RMSE consistently decreases with increasing sample size, confirming the asymptotic consistency of the estimators. This trend is clearly illustrated in Figure \ref{Figure 3}, which shows the RMSE plots for each estimation method across different sample sizes. However, a comparative analysis reveals that among the considered parameter settings, Bayesian estimation yielded lower RMSEs values across all sample sizes.  
	
	\section{KTD Regression model} \label{reg}
	In this section, we introduce parametric KTD regression model. Although the quantile function is available in closed form through the Lambert \(W\) function, the mean and variance do not provide a comparably simple basis for regression parameterization. We therefore re-parameterize \(b\) through its median, which is directly obtained from the quantile function as:
	\[b = \frac{\log(0.5)}{\log\left( 1-\left( 1-e^{\fm} \right)^\a \right)}\]
	Let \(Y\) represent a random variable following the KTD regression framework. By introducing the re-parameterized expression for \(b\) into eq. \eqref{ktcdf}, the CDF is defined as:
	\begin{eqnarray}\label{ktrcdf}
		\F(\y; \mu, \a, \th) = 1-\left( 1-\left( 1-e^{\fy} \right)^\a \right)^{\frac{\log(0.5)}{\log\left( 1-\left( 1-e^{\fm} \right)^\a \right)}},~~~ y>0
	\end{eqnarray}
	and the corresponding PDF is given by:
	\begin{eqnarray}\label{ktrpdf}
		\f(\y; \mu, \a, \th) = \frac{\a \th \log(0.5) (e^{\th \y} -1) e^{\fy} \left(1-e^{\fy} \right)^{\a-1} \left(  1-\left(1-e^{\fy} \right)^{\a} \right)^{\frac{\log(0.5)}{\log\left( 1-\left( 1-e^{\fm} \right)^\a \right)}-1}}{\log\left( 1-\left( 1-e^{\fm} \right)^\a \right)}~
	\end{eqnarray}
	Similar to eq. \eqref{quantile}, we can obtain the quantile function of regression model as:
	\begin{eqnarray}\label{ktrqf}
		y_{p} (p; \mu, \a, \th) = \frac{1}{\th}~ ln \left[ -W_{-1} \left(  \left\{ \left(1-(1-p)^{\frac{\log\left( 1-\left( 1-e^{\fm} \right)^\a \right) }{\log(0.5)}} \right)^{1/a}-1\right\} /e \right) \right],~~~ 0<p<1~
	\end{eqnarray}
	Now, consider observations \(Y_1,Y_2,\ldots,Y_n\), where each response \(Y_i\) follows the KTD distribution given in Equation~(\ref{ktrpdf}), with median \(\mu_i\) and unknown parameters \(\Theta=(a,\theta)\). We adopt a working-independence formulation in which the conditional marginal distribution of \(Y_i\) is characterized through its median \(\mu_i\), without explicitly specifying a temporal dependence structure. The regression framework models the median \(\mu_i\) as a function of explanatory variables through
	\begin{equation} \label{regr}
		\eta_i=g(\mu_i)=\sum_{m=1}^{k}z_{im}\beta_m
		=z_i^T\boldsymbol{\beta}.
	\end{equation}
	Here, \(\boldsymbol{\beta}=(\beta_1,\ldots,\beta_k)^T\in\mathbb{R}^k\) denotes the vector of regression coefficients, with \(k<n\), and \(\eta=(\eta_1,\ldots,\eta_n)^T\) denotes the vector of linear predictors. The vector \(z_i^T\) contains the covariate values for the \(i\)th observation. The link function \(g(\cdot)\) is assumed to be strictly monotone and twice continuously differentiable, with its range and domain determined by the support of the median parameter.

	\subsection{Likelihood inference} \label{cmle}
	Let $Y_i \sim \text{KTD}(\mu_i, \Theta), ~~i = 1,2,\cdots,n$, with observed values $y_i$ and co-variates \(z_{i}\). Then, the log-likelihood of KTD regression model is given by: 
	\begin{eqnarray} \label{CMLE}
		\ell_i( \mu_i , \a, \th) &=&  \log(\a \th) + \log(-\log(0.5)) - \log\left\{-\log\left( 1-\left( 1-e^{\fmt} \right)^\a \right) \right\} + \log(e^{\th \y_i} -1) +  \fyt  \nonumber \\ &&  + (\a-1) \log\left(1-e^{\fyt} \right)  + \left(\psi_i -1\right) \log \left(1- \left(1-e^{\fyt} \right)^\a \right)   
	\end{eqnarray}
	where, $\psi_i =  \frac{\log(0.5)}{\log\left( 1-\left( 1-e^{\fmt} \right)^\a \right)}$.
	
	To obtain the MLEs denoted by $\hat{a}$, $\hat{\theta}$, $\hat{\boldsymbol{\beta}} $, we need to simultaneously solve a system of non-linear score equations. These expressions correspond to the first-order partial derivatives of the log-likelihood function and are defined as follows:
	\begin{align*}
		U_{\a} = \sum_{i=1}^{n} \frac{\partial 	\ell_i}{\partial \a}, ~~
		U_{\theta} = \sum_{i=1}^{n} \frac{\partial 	\ell_i}{\partial \theta},~~
		U_{\beta_j} = \sum_{i=1}^{n} \frac{\partial	\ell_i}{\partial \mu_i} \frac{d \mu_i}{d \eta_i} \frac{\partial \eta_i}{\partial \beta_j},
	\end{align*}
	These derivatives constitute the elements of the complete score vector. In the case of the regression coefficients \(\beta _{j}\), the score component explicitly utilizes the chain rule to map the effect of the covariates through the linear predictor \(\eta _{i}\) and the inverse link function to the individual medians \(\mu _{i}\). Since the score equations are nonlinear and analytically intractable, closed-form expressions for the parameter estimators cannot be obtained. Therefore, the resulting optimization problem is solved iteratively using numerical maximization methods, such as the Newton--Raphson or quasi-Newton algorithms, to obtain the parameter estimates by maximizing the log-likelihood function until convergence.

	\subsection{Diagnostic Analysis} \label{diagnostic} 
	Model adequacy is assessed using Cox--Snell residuals \cite{cox1968general} and normalized quantile residuals \cite{dunn1996randomized}. For the fitted conditional distribution \(F_t\), the Cox--Snell residual is defined as
	\[r_t=-\log\left\{1-F_t(y_t| \mu_t)\right\}\]
	Under a correctly specified continuous model, the Cox--Snell residuals follow a standard exponential distribution. Their adequacy is therefore assessed graphically using a P--P plot, which compares the empirical distribution of the residuals with the theoretical \(Exp(1)\) distribution. A close agreement between the empirical and theoretical probabilities indicates an adequate distributional fit. In addition, normalized quantile residuals are obtained as
	\[z_t=\Phi^{-1}\left\{F_t(y_t | \mu_t)\right\}\]
	where \(\Phi^{-1}\) denotes the quantile function of the standard normal distribution. Under correct model specification, \(z_t\) follows \(N(0,1)\). Their distributional adequacy is assessed using a normal Q--Q plot, in which the ordered \(z_t\) values are compared with the corresponding theoretical standard normal quantiles. For illustration, both the P--P plot of the Cox--Snell residuals and the Q--Q plot of the normalized quantile residuals are presented for a representative grid. For the analysis of the complete gridded dataset, where graphical assessment for each location is impractical, the Cramér--von Mises (CvM) goodness-of-fit test is used as an automated diagnostic measure. The test evaluates the discrepancy between the empirical distribution of the normalized quantile residuals and the standard normal distribution:
	$$
	W^{2}=n\int_{-\infty}^{\infty}
	\left[F_n(z)-\Phi(z)\right]^2d\Phi(z),
	$$
	where \(n\) is the number of temporal observations. The null hypothesis (\(H_{0}\)) states that the empirical \(z\)-scores follow a standard normal profile. A computed \(p\)-value below the significance threshold (\(p < 0.05\)) rejects the null hypothesis, indicating model mis-specification. Conversely, a \(p\)-value \(\ge 0.05\) confirms that the empirical residuals do not deviate significantly from the target Gaussian curve, demonstrating a successful distributional fit.
	
	\subsection{Simulation Study}
	The finite-sample performance of the MLE developed in subsection~(\ref{cmle}) is examined through Monte Carlo experiments for the proposed KTD regression model. For each setting, 1000 independent datasets are generated with sample sizes $n=50,100,200$ and $500$. Three cases are considered, corresponding to models with one, two, and three regressors in addition to the intercept $\beta_0$. The parameter settings are given by:
	\begin{itemize}
		\item  $a=1.2$, $\th=2.8$, $\beta_0=-0.7$, $\beta_1=-1.4$
		
		\item
		$a = 2.5$, $\th = 0.8$, $\beta_0 = 1.6$, $\beta_1 = 0.5$, $\beta_2 = -1.7$.
		
		\item 
		$a = 0.9$, $\th = 1.5$, $\beta_0 = 0.4$, $\beta_1 = -0.6$, $\beta_2 = 0.8$, $\beta_3 = -1.2$.
	\end{itemize}
	
	\begin{table}[!ht]
		\centering
		\small
		\caption{Mean estimates, Bias and MSE of the KTD regression model.}
		\resizebox{\textwidth}{!}{
			\begin{tabular}{@{}lcccccccccccc@{}}
				\toprule
				& \multicolumn{4}{c}{Mean} & \multicolumn{4}{c}{Bias} & \multicolumn{4}{c}{MSE} \\
				\cmidrule(lr){2-5} \cmidrule(lr){6-9} \cmidrule(lr){10-13}
				{} & 50 & 100 & 200 & 500 & 50 & 100 & 200 & 500 & 50 & 100 & 200 & 500 \\
				
				\midrule \multicolumn{13}{c}{True values: $a=1.2$, $\th=2.8$, $\beta_0=-0.7$, $\beta_1=-1.4$} \\ \midrule $\hat{a}$
				
				& 1.2478 & 1.2273 & 1.2101 & 1.2046 & 0.0478 & 0.0273 & 0.0101 & 0.0046 & 0.0358 & 0.0155 & 0.0071 & 0.0029 \\ $\hat{\th}$ & 2.8906 & 2.8964 & 2.8522 & 2.8216 & 0.0906 & 0.0964 & 0.0522 & 0.0216 & 0.2376 & 0.2309 & 0.2281 & 0.2178 \\ $\hat{\beta}_0$ & -0.6992 & -0.6941 & -0.7020 & -0.6984 & 0.0008 & 0.0059 & -0.0020 & 0.0016 & 0.0215 & 0.0129 & 0.0057 & 0.0023 \\ $\hat{\beta}_1$ & -1.4077 & -1.4092 & -1.3993 & -1.4013 & -0.0077 & -0.0092 & 0.0007 & -0.0013 & 0.0583 & 0.0338 & 0.0166 & 0.0061 \\

				\midrule
				\multicolumn{13}{c}{True values: $a=2.5$, $\th=0.8$, $\beta_0=1.6$, $\beta_1=0.5$, $\beta_2=-1.7$} \\
				\midrule
				
				$\hat{a}$
				& 2.6628 & 2.6201 & 2.5928 & 2.5552
				& 0.1628 & 0.1201 & 0.0928 & 0.0552
				& 0.1003 & 0.0694 & 0.0485 & 0.0249 \\
				
				$\hat{\th}$
				& 0.6958 & 0.6945 & 0.7510 & 0.7595
				& -0.1042 & -0.1055 & -0.0490 & -0.0405
				& 0.1988 & 0.1846 & 0.1653 & 0.1250 \\
				
				$\hat{\beta}_0$
				& 1.6224 & 1.6225 & 1.6042 & 1.6055
				& 0.0224 & 0.0225 & 0.0042 & 0.0055
				& 0.0849 & 0.0554 & 0.0311 & 0.0118 \\
				
				$\hat{\beta}_1$
				& 0.5096 & 0.5103 & 0.5074 & 0.5002
				& 0.0096 & 0.0103 & 0.0074 & 0.0002
				& 0.1521 & 0.0857 & 0.0283 & 0.0099 \\
				
				$\hat{\beta}_2$
				& -1.7402 & -1.7344 & -1.7125 & -1.7061
				& -0.0402 & -0.0344 & -0.0125 & -0.0061
				& 0.1943 & 0.0795 & 0.0464 & 0.0168 \\

				\midrule
				\multicolumn{13}{c}{True values: $a=0.9$, $\th=1.5$, $\beta_0=0.4$, $\beta_1=-0.6$, $\beta_2=0.8$, $\beta_3=-1.2$} \\
				\midrule
				
				$\hat{a}$
				& 0.9542 & 0.9220 & 0.9036 & 0.9039
				& 0.0542 & 0.0220 & 0.0036 & 0.0039
				& 0.0294 & 0.0159 & 0.0085 & 0.0041 \\
				
				$\hat{\th}$
				& 1.6268 & 1.5869 & 1.5666 & 1.5112
				& 0.1268 & 0.0869 & 0.0666 & 0.0112
				& 0.2083 & 0.1847 & 0.1535 & 0.1073 \\
				
				$\hat{\beta}_0$
				& 0.4106 & 0.4172 & 0.4098 & 0.4074
				& 0.0106 & 0.0172 & 0.0098 & 0.0074
				& 0.1154 & 0.0682 & 0.0475 & 0.0205 \\
				
				$\hat{\beta}_1$
				& -0.5912 & -0.6157 & -0.6142 & -0.6081
				& 0.0088 & -0.0157 & -0.0142 & -0.0081
				& 0.1392 & 0.0944 & 0.0483 & 0.0216 \\
				
				$\hat{\beta}_2$
				& 0.8272 & 0.8146 & 0.8140 & 0.8058
				& 0.0272 & 0.0146 & 0.0140 & 0.0058
				& 0.1482 & 0.0889 & 0.0519 & 0.0214 \\
				
				$\hat{\beta}_3$
				& -1.2436 & -1.2171 & -1.2107 & -1.2044
				& -0.0436 & -0.0171 & -0.0107 & -0.0044
				& 0.1731 & 0.1024 & 0.0602 & 0.0248 \\
				\bottomrule
		\end{tabular}}
		\label{sim1}
	\end{table}
	The simulation results for the three parameter settings are reported in Table~\ref{sim1}. Overall, the mean estimates exhibit decreasing deviations from their corresponding true parameter values as the sample size increases. For all three parameter settings, the absolute biases and MSEs generally decline as the sample size increases. This reduction indicates that the parameter estimates become more accurate and stable for larger samples. For the first setting, the estimates of $a$, $\beta_0$, and $\beta_1$ approach their respective true values rapidly, while $\theta$ exhibits a comparatively larger bias and MSE, although both decrease with increasing sample size. In the second setting, the same pattern is observed, with the estimates of the regression coefficients remaining close to their true values even for smaller samples. The third setting, which includes three regressors, also shows satisfactory parameter recovery, with the mean estimates approaching the true values and the corresponding biases and MSEs decreasing as $n$ increases. These results indicate that the MLE provides stable and accurate parameter estimates across the considered model specifications.
	
	\section{Statistical Data Analysis} \label{data}
	We consider monthly gridded rainfall data for the Northwest Himalayas (NWH), covering Jammu \& Kashmir (JK), Himachal Pradesh (HP), and Uttarakhand (UK), over the period 2016--2025. The data are obtained from the India Meteorological Department (IMD) (\small\url{https://www.imdpune.gov.in/cmpg/Griddata/Rainfall_25_NetCDF.html}) \normalsize at a spatial resolution of \(0.25^\circ\times0.25^\circ\). The data set consists of 525 grid cells and used as the response variable.
	\begin{figure}[!ht]
		\begin{minipage}[t]{0.33\textwidth}
			\centering
			\includegraphics[width=1\textwidth]{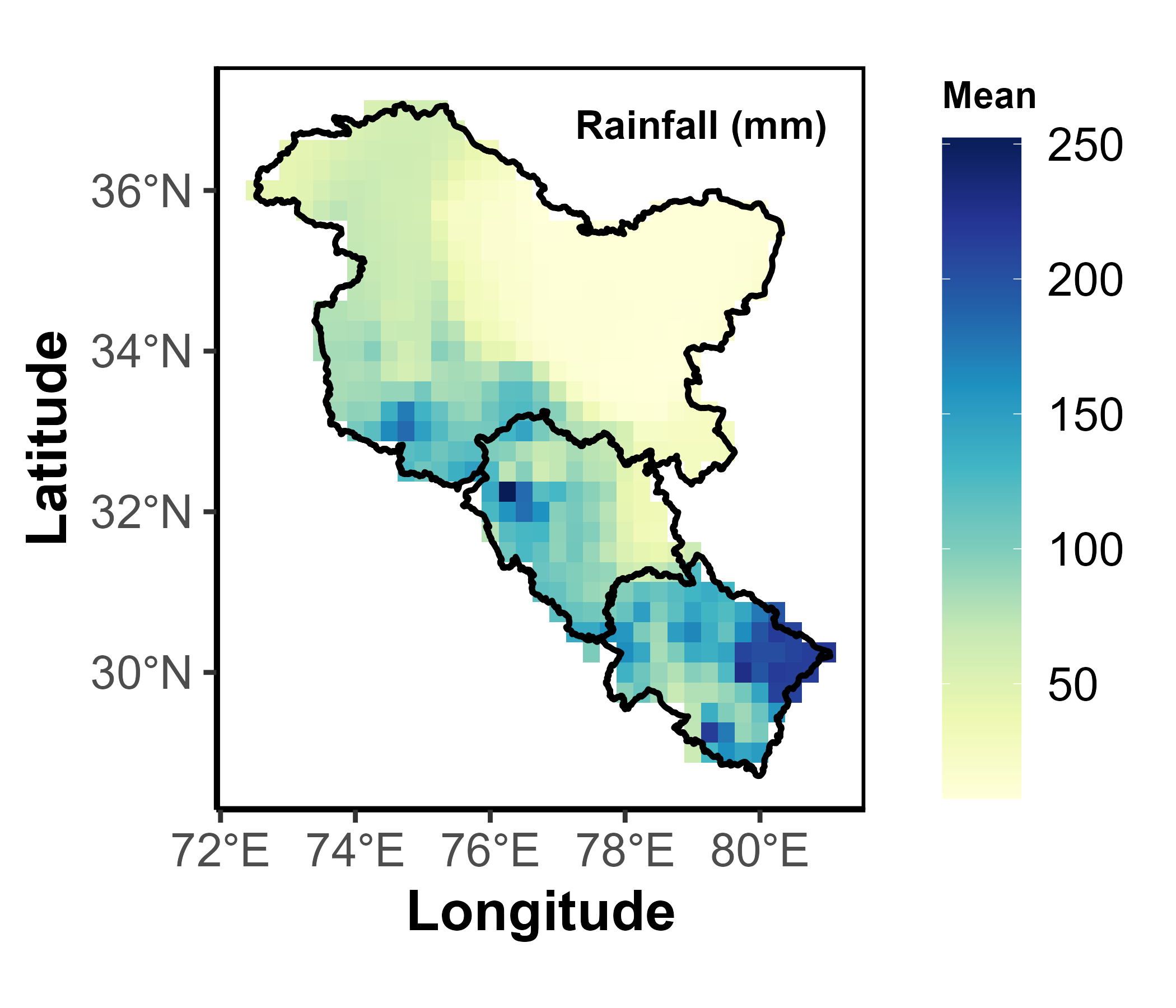}
			\caption*{(a) }
		\end{minipage}\hfill
		\begin{minipage}[t]{0.33\textwidth}
			\centering
			\includegraphics[width=1\textwidth]{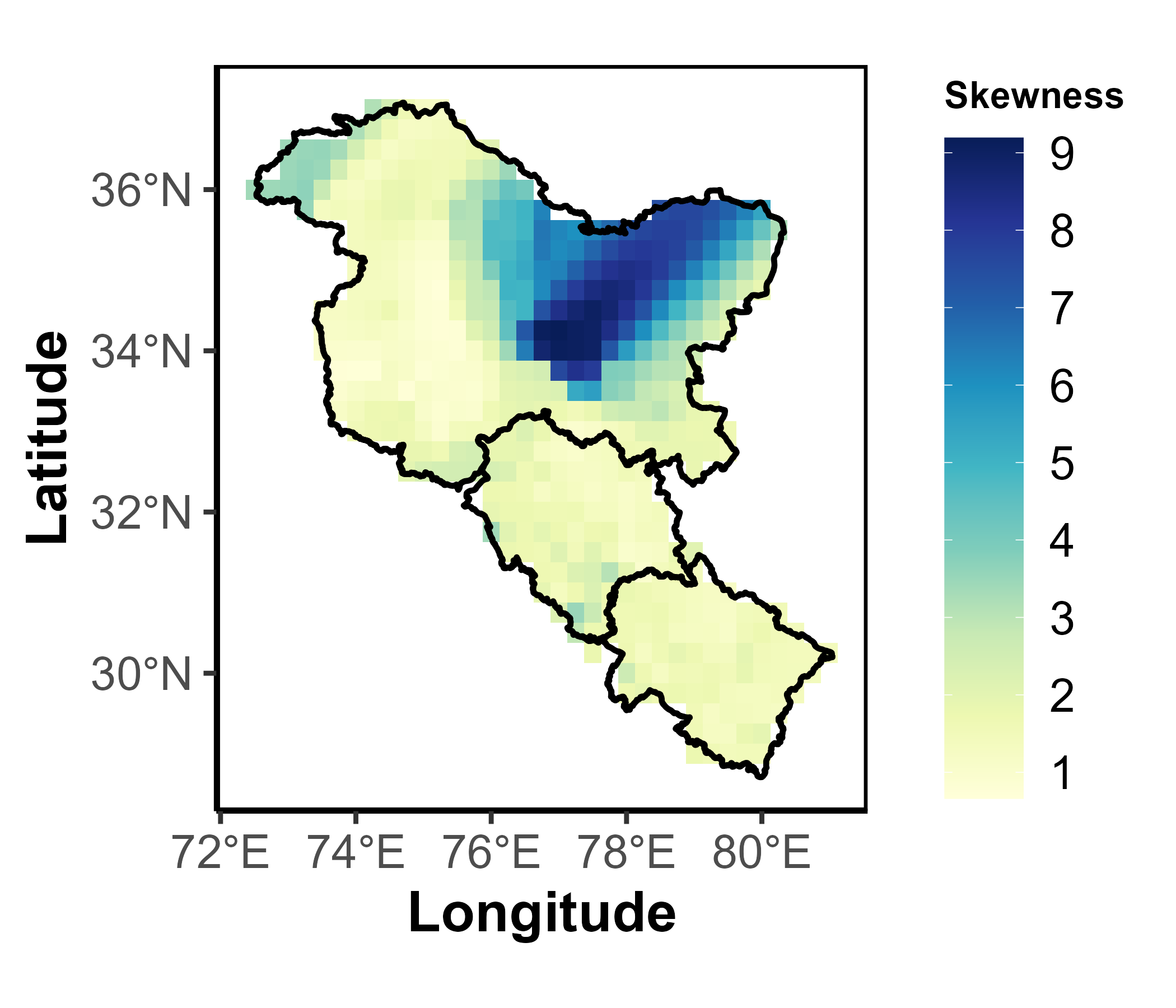}
			\caption*{(b) }
		\end{minipage}\hfill
		\begin{minipage}[t]{0.33\textwidth}
			\centering
			\includegraphics[width=1\textwidth]{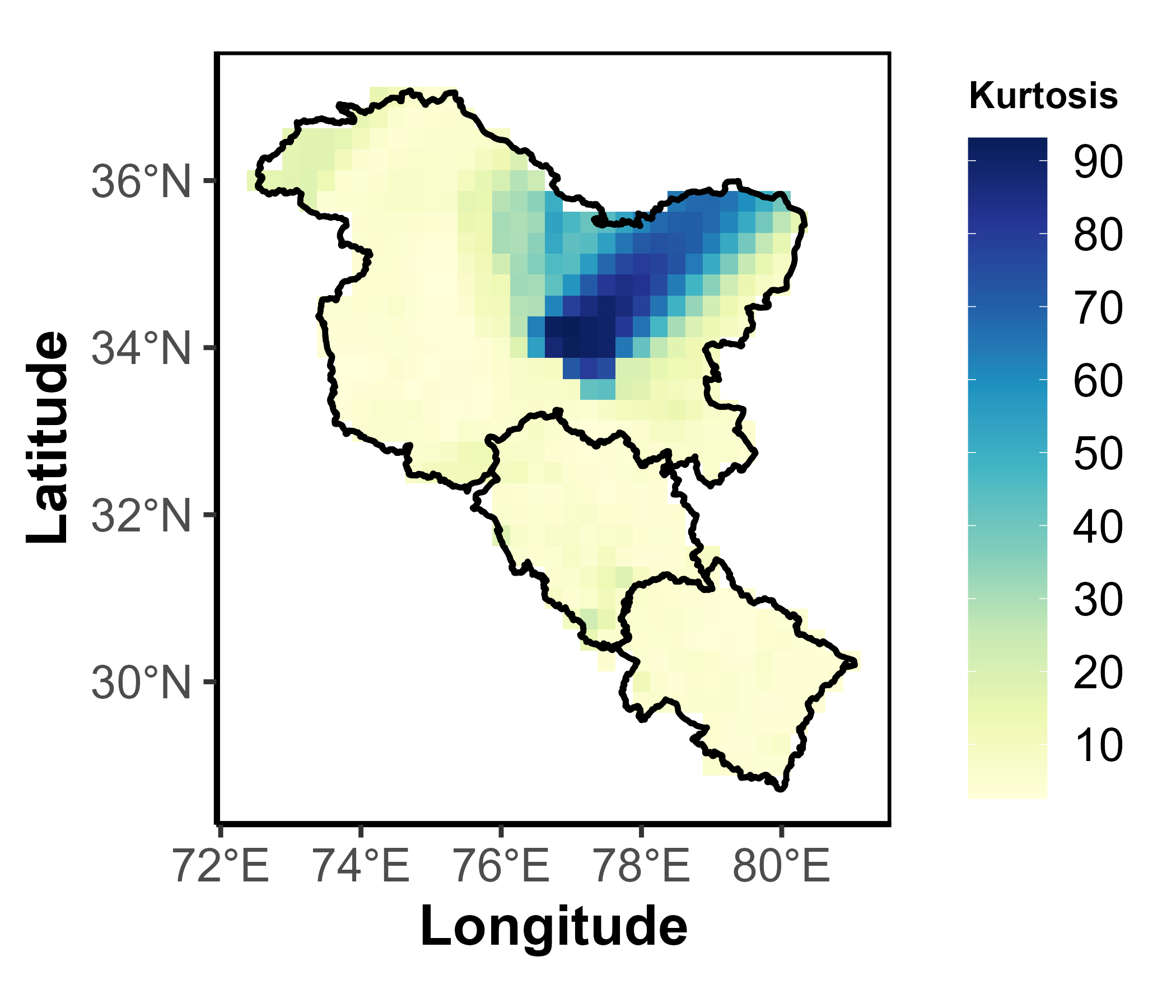}
			\caption*{(c) }
		\end{minipage}
		\caption{ Descriptive Statistics of NWH Rainfall for the time-period 2016--2025: (a) Mean, (b) Skewness and (c) Kurtosis.}
		\label{descriptive}
	\end{figure}
	Figure \ref{descriptive} illustrates the spatial distribution of the three statistical moments (mean, skewness and kurtosis) derived from the raw monthly gridded rainfall series over the 525 grid cells for the 2016–2025 time period. The mean spatial rainfall pattern (Figure \ref{descriptive}a) indicates the presence of different hydroclimatic zones due to the complex orography of the NWH. The higher mean values are clustered along the southern foothills and mountain slopes of HP and UK. Conversely, mean rainfall totals are severely depressed in the trans-Himalayan rain shadow areas of northern JK. The skewness map (Figure~\ref{descriptive}b) shows positive values across all 525 locations, indicating that the monthly rainfall distributions are consistently right-skewed. The pronounced skewness suggests that rainfall totals are characterized by a concentration of moderate observations accompanied by occasional large rainfall events. The kurtosis distribution (Figure~\ref{descriptive}c) further shows the similar pattern and indicates substantial departures from normality, with elevated values across the region and particularly pronounced levels in the upper eastern region. The combined evidence of positive skewness and high kurtosis highlights the non-normal and heterogeneous nature of the rainfall distributions.  
	
	To account for seasonal variation, annual and semi-annual harmonic terms, $\sin(2 \pi t/12)$, $\cos(2 \pi t/12)$, $\sin(2 \pi t/6)$, and $\cos(2 \pi t/6)$, are included as regressors, following the purely seasonal specification considered for the generalized gamma model by \cite{perreault2025distributional}. 
	\begin{figure}[!ht]
		\begin{minipage}[t]{0.33\textwidth}
			\centering
			\includegraphics[width=1\textwidth]{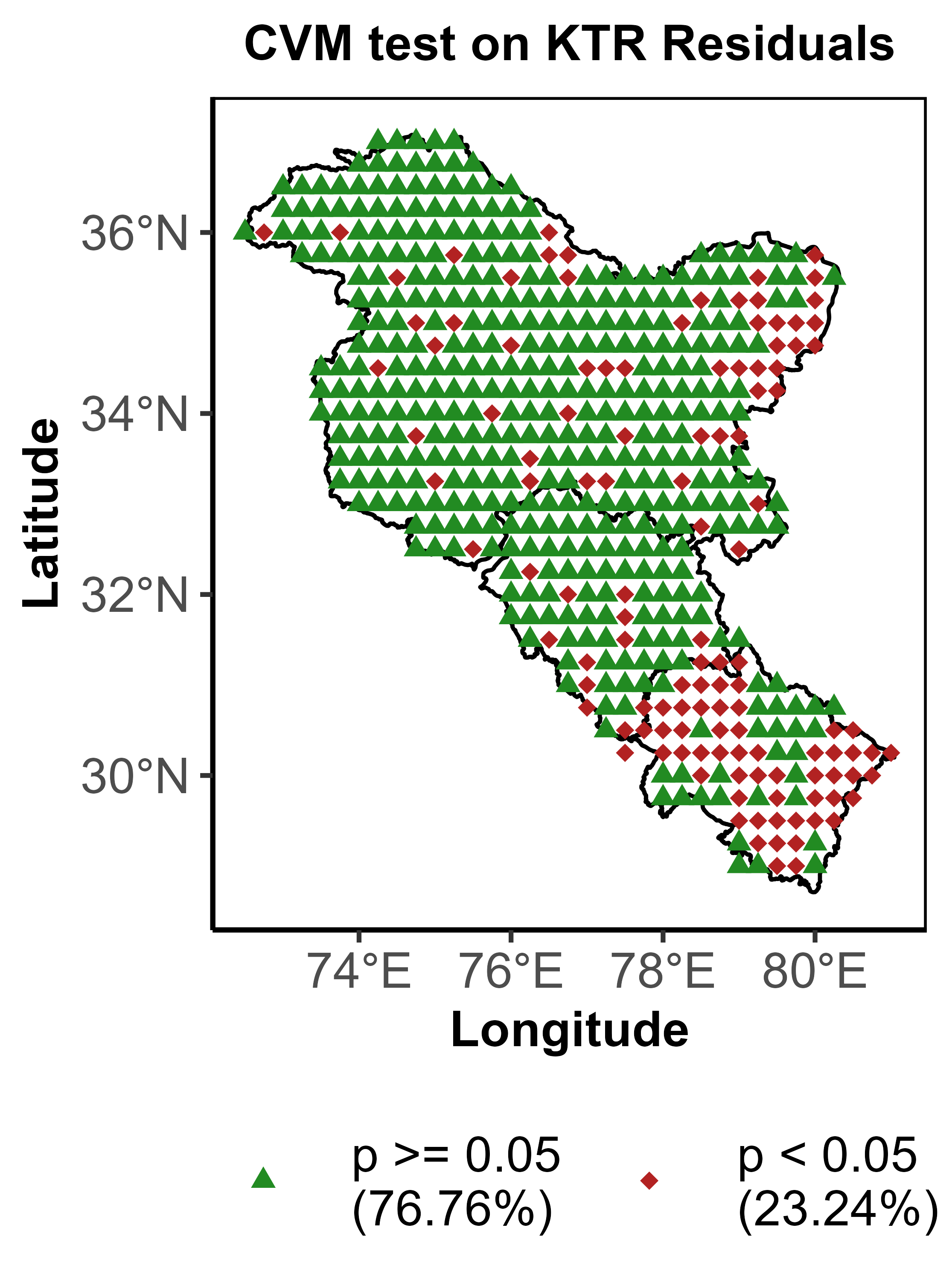}
			\caption*{(a) }
		\end{minipage}\hfill
		\begin{minipage}[t]{0.33\textwidth}
			\centering
			\includegraphics[width=1\textwidth]{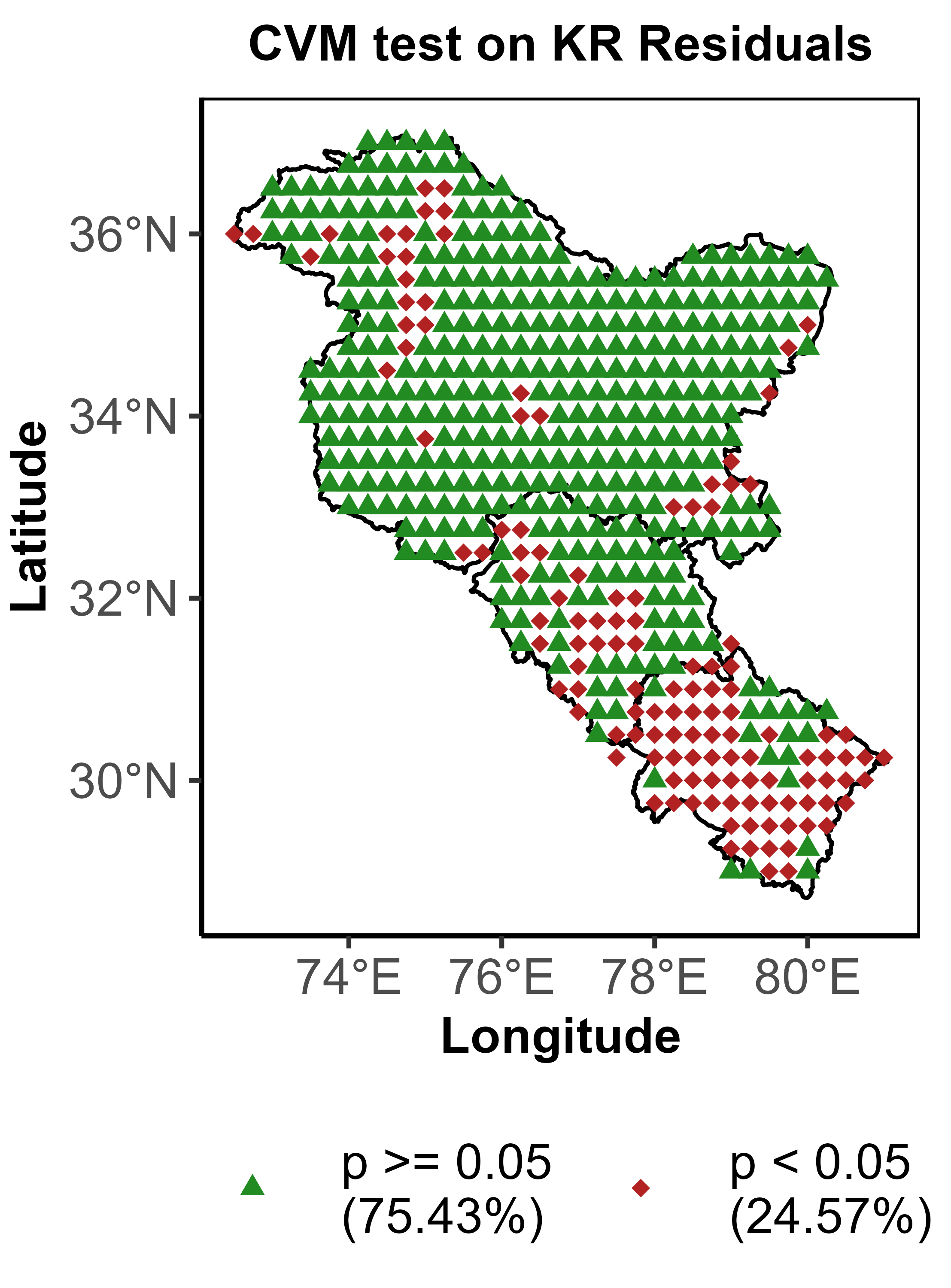}
			\caption*{(b) }
		\end{minipage}\hfill
		\begin{minipage}[t]{0.33\textwidth}
			\centering
			\includegraphics[width=1\textwidth]{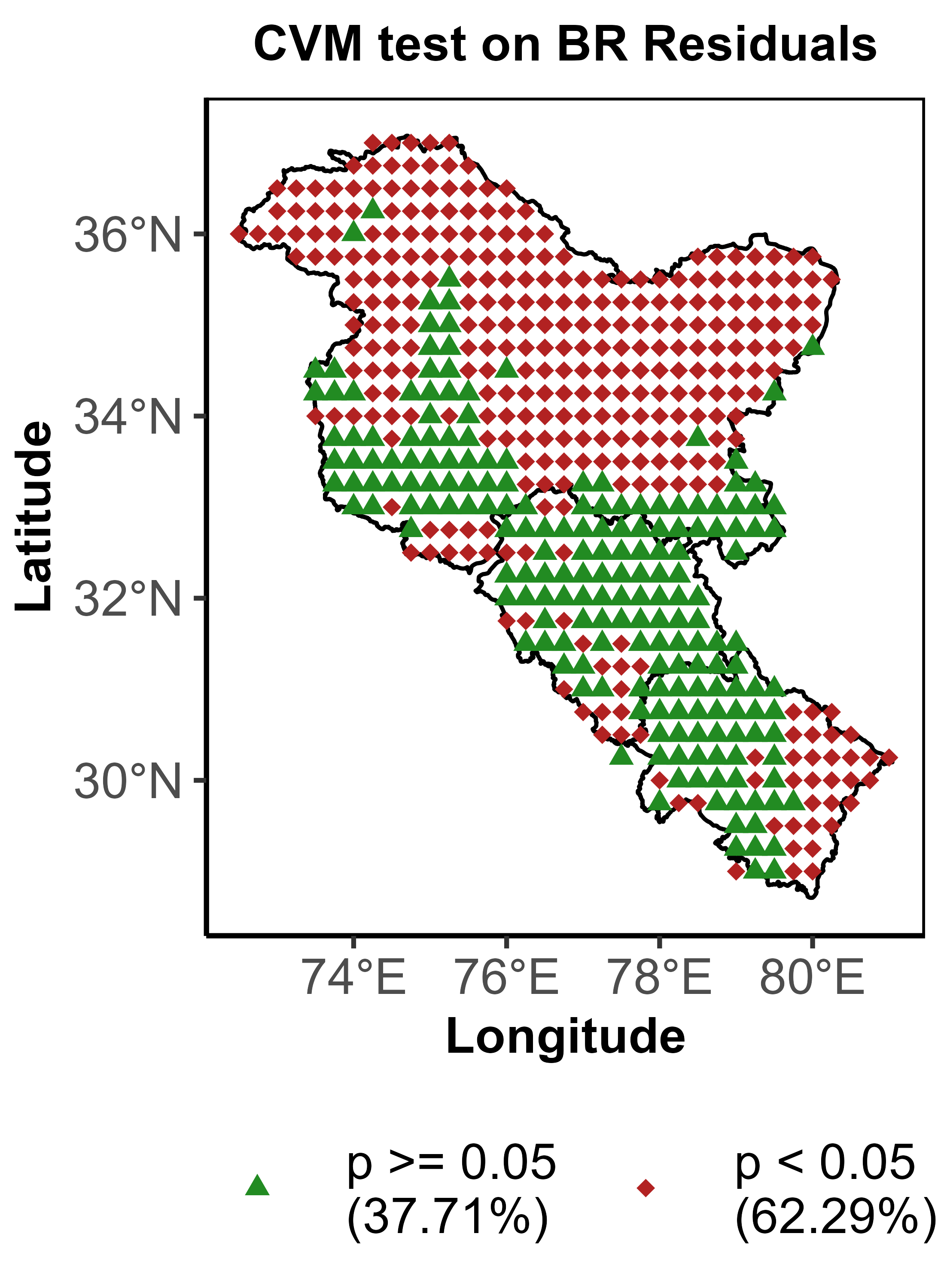}
			\caption*{(c) }
		\end{minipage}
		\caption{Spatial mapping of residual diagnostics via Cram\'{e}r-von Mises goodness-of-fit tests for the quantile residuals: (a) Kumaraswamy-Teissier Regression, (b) Kumaraswamy Regression and (c) Beta Regression.}
		\label{Residuals}
	\end{figure}
	Accordingly, the regression structure is specified through the logit link as
	\begin{equation}\label{logit}
		\log \left( \frac{\mu_{t}}{1-\mu_{t}} \right) = \beta_0 + \sum_{m \in \{6, 12\}} \left[ \beta_{\sin(2 \pi t/m)} \sin\left(\frac{2\pi t}{m}\right) + \beta_{\cos(2 \pi t/m)} \cos\left(\frac{2\pi t}{m}\right) \right]
	\end{equation}
	
	Using equation~(\ref{logit}), the rainfall series at each grid is modelled using three regression models: Kumaraswamy-Teissier Regression (KTR), Kumaraswamy Regression (KR), and Beta Regression (BR) models. Since, the KR and BR models are defined on the unit interval, the rainfall observations are linearly transformed to $(0,1)$ prior to model fitting. To maintain a common response scale across the three models, the same transformed rainfall series is used for each model. Before selecting the best-fitting model at each grid, the fitted models are subjected to the diagnostic procedure described in Section~\ref{diagnostic}. In particular, the Cram'{e}r--von Mises (CvM) goodness-of-fit test is applied to the quantile residuals to assess their distributional adequacy. For the grid-wise analysis, the CvM test is used as an automated diagnostic screening criterion, with models having \(p\ge0.05\) retained for subsequent comparison. Figure~\ref{Residuals} shows the spatial distribution of the CvM test results for the three models. The KTR model satisfies the diagnostic criterion at 403 of the 525 grids (76.76\%), while the corresponding numbers for KR and BR are 396 (75.43\%) and 198 (37.71\%), respectively. The KTR and KR models provide acceptable residual fits over a large part of JK and HP, including areas characterized by pronounced skewness and kurtosis, whereas the BR model satisfies the diagnostic criterion at a smaller proportion of the grids covering western UK and HP.
	
	\begin{table}[!ht]
		\centering
		\caption{Parameter Estimates and Distributional Regression Diagnostics for Rainfall data at Grid Location: 37$^{\circ}$N, 74.25$^{\circ}$E}
		\label{tab:grid_521_details}
		\small 
		\begin{tabular}{lccccccccc}
			\hline
			\textbf{Parameter} & \multicolumn{3}{c}{\textbf{KTR Model}} & \multicolumn{3}{c}{\textbf{KR Model}} & \multicolumn{3}{c}{\textbf{BR Model}} \\ \cline{2-10} 
			& \textbf{Estimate} & \textbf{SE} & \textbf{$p$-value} & \textbf{Estimate} & \textbf{SE} & \textbf{$p$-value} & \textbf{Estimate} & \textbf{SE} & \textbf{$p$-value} \\ \hline
			$a$                 & 0.4685  & 0.0398 & $<0.001$ & 0.8156  & 0.0726 & $<0.001$ & —       & —      & —        \\
			$\theta$           & 1.4136  & 0.0448 & $<0.001$ & —       & —      & —        & —       & —      & —        \\
			$\phi$              & —       & —      & —        & —       & —      & —        & 3.4071  & 0.4585 & $<0.001$ \\
			$\beta_0$            & -2.4387 & 0.1140 & $<0.001$ & -2.3323 & 0.1316 & $<0.001$ & -1.5536 & 0.1127 & $<0.001$ \\
			$\beta_{\sin(2 \pi t/12)}$  & -0.2946 & 0.1455 & 0.0428   & -0.4179 & 0.1715 & 0.0148   & -0.0437 & 0.1294 & 0.7353   \\
			$\beta_{\cos(2 \pi t/12)}$  & -0.6250 & 0.1444 & $<0.001$ & -0.9917 & 0.1704 & $<0.001$ & -0.5590 & 0.1321 & $<0.001$ \\
			$\beta_{\sin(2 \pi t/6)}$   & 0.5574  & 0.1356 & $<0.001$ & 0.8140  & 0.1587 & $<0.001$ & 0.5387  & 0.1339 & $<0.001$ \\
			$\beta_{\cos(2 \pi t/6)}$   & -0.3218 & 0.1552 & 0.0381   & -0.1054 & 0.1881 & 0.5754   & -0.0726 & 0.1277 & 0.5695   \\ \hline
			\textbf{Diagnostics} &        &        &          &         &        &          &         &        &          \\
			
			CvM $p$-value       & \multicolumn{3}{c}{0.3508}      & \multicolumn{3}{c}{0.2616}      & \multicolumn{3}{c}{0.0006}      \\
			Log-Likelihood      & \multicolumn{3}{c}{149.7505}    & \multicolumn{3}{c}{140.8173}    & \multicolumn{3}{c}{118.2318}    \\
			AIC                 & \multicolumn{3}{c}{{-285.5011}} & \multicolumn{3}{c}{-269.6346}  & \multicolumn{3}{c}{-224.4637}  \\
			BIC                 & \multicolumn{3}{c}{{-265.9886}} & \multicolumn{3}{c}{-252.9097}  & \multicolumn{3}{c}{-207.7387}  \\ \hline
		\end{tabular}
	\end{table}
	Since it is impractical to report parameter estimates and diagnostic measures for all 525 grids in tabular form, Table~\ref{tab:grid_521_details} presents detailed results for the grid located at 37$^{\circ}$N, 74.25$^{\circ}$E. 
	\begin{figure}[!ht]
		\begin{center}
			\includegraphics[width=1.0\textwidth]{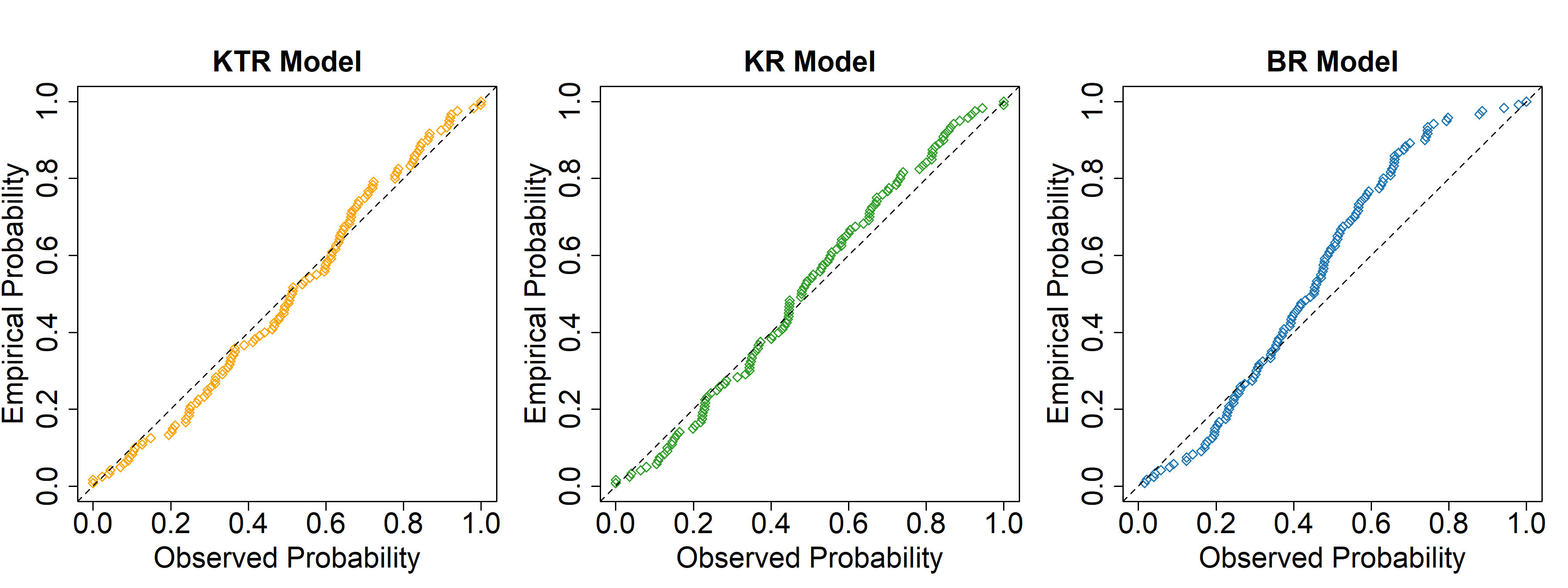}
		\end{center}
		\caption{P-P plots of the Cox-Snell residuals of the KTR, KR, and BR regression models for rainfall data at grid location: 37$^{\circ}$N, 74.25$^{\circ}$E.}
		\label{pp-map}
	\end{figure} 
	\begin{figure}[!ht]
		\begin{center}
			\includegraphics[width=1.0\textwidth]{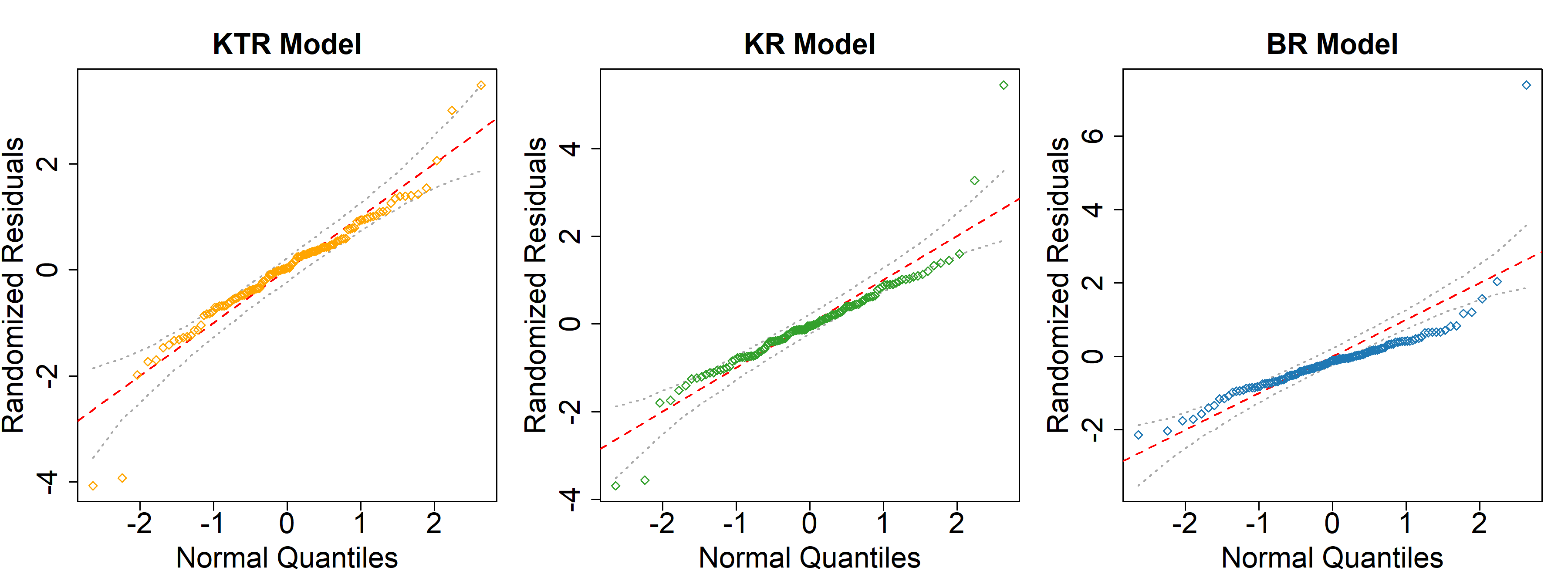}
		\end{center}
		\caption{Q-Q plots of the randomized quantile residuals of the KTR, KR, and BR regression models for rainfall data at grid location: 37$^{\circ}$N, 74.25$^{\circ}$E.}
		\label{qq-map}
	\end{figure}
	The table provides the parameter estimates, standard errors, corresponding (p)-values, CvM (p)-values, log-likelihood, Akaike Information Criterion (AIC), and Bayesian Information Criterion (BIC) for the three models. The AIC and BIC are defined as
	\begin{equation*}
		\mathrm{AIC}=2k-2\ell(\hat{\Theta}), \qquad
		\mathrm{BIC}=k\log(n)-2\ell(\hat{\Theta}),
	\end{equation*}
	where \(\ell(\hat{\Theta})\) denotes the maximized log-likelihood, (k) is the number of estimated parameters, and (n) is the sample size. Both criteria balance model fit against model complexity, with smaller values indicating a more favourable trade-off between the two. For this grid, both KTR and KR satisfy the CvM criterion, whereas BR does not. Among the three models, KTR attains the largest log-likelihood and the smallest AIC and BIC, indicating a better fit relative to the competing models after accounting for model complexity. The graphical diagnostics are consistent with the numerical comparison, with the P--P plot of the Cox--Snell residuals in Figure~\ref{pp-map} and the normal Q--Q plot of the randomized quantile residuals in Figure~\ref{qq-map} showing closer agreement with their respective theoretical distributions for KTR at this grid. The significance of the harmonic coefficients further indicates that both annual and semi-annual components contribute to the seasonal variation in rainfall at this location.
	
	The same procedure is applied independently to all 525 grids. After filtering the candidate models using the CvM diagnostic, the model with the largest log-likelihood is selected at each grid. The resulting spatial distribution is shown in Figure~\ref{ll-map}. Approximately 57\% of the grids are selected as KTR, 14\% as KR, and 22\% as BR, while the remaining 32 grids are not adequately represented by any of the three models. 
	\begin{figure}[!ht]
		\begin{center}
			\includegraphics[width=0.6\textwidth]{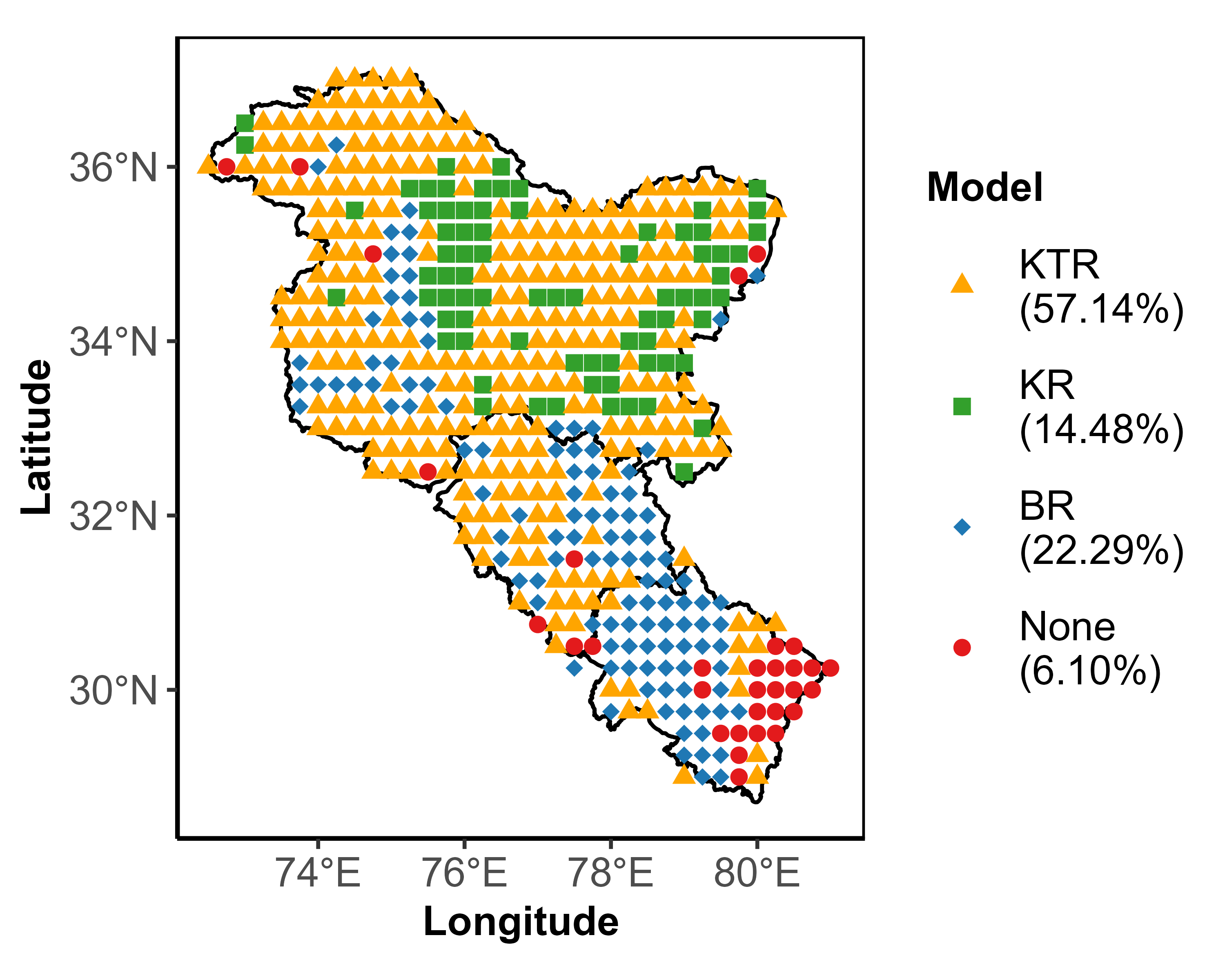}
		\end{center}
		\caption{Spatial distribution of the optimal model selection based on maximum log-likelihood values pre-filtered by the Cram\'{e}r-von Mises diagnostic test ($p \geq 0.05$). }
		\label{ll-map}
	\end{figure} 
	The KTR model is selected across a substantial part of western HP and JK, including regions exhibiting pronounced skewness and kurtosis, whereas BR is selected over much of UK and eastern HP, and KR is selected mainly over the western part of JK. The widespread selection of KTR, together with its relatively high proportion of grids satisfying the CvM diagnostic, demonstrates its utility in representing the heterogeneous and non-normal rainfall distributions observed across the NWH. At the same time, the spatially varying selection of KR and BR indicates that the distributional characteristics of rainfall differ across the region and that alternative regression models remain relevant for specific areas.

	\section{Conclusion} \label{conclusion}
	In this study, we introduce the Kumaraswamy--Teissier distribution (KTD), a three-parameter extension obtained by applying the Kumaraswamy-$G$ construction to the Teissier distribution. The additional shape parameters of the Kumaraswamy generator, together with the scale parameter, provide a flexible distribution for positive-valued observations. The proposed distribution can accommodate different density and hazard-rate shapes, including increasing, decreasing, and bathtub-type hazard-rate behaviour. Several statistical properties of the KTD are derived, including its moments, moment generating function, quantile function, and order-statistic distributions. For parameter estimation, maximum likelihood, maximum product spacing, and Bayesian methods are considered. The simulation results show that the estimated parameters approach their true values as the sample size increases, while the RMSE decreases across the considered settings. Among the three estimation approaches, the Bayesian method provides the most favourable performance in terms of mean estimates and RMSE.
	
	We further develop a median-based KTD regression model and demonstrate its finite-sample performance through simulation, where the bias and MSE decrease as the sample size increases. The practical utility of the proposed regression model is illustrated using NWH rainfall data with annual and semi-annual harmonic regressors. Compared with Kumaraswamy and beta regression, the KTD regression satisfies the CvM diagnostic at 76.76\% of the 525 grids and is selected for approximately 57\% of the grids after diagnostic filtering and maximum log-likelihood. Overall, the three models provide an adequate representation for 493 of the 525 grids, leaving 32 grids unexplained by the considered candidates. The dominant spatial footprint of the KTD framework, particularly across regions exhibiting severe non-normal skewness and high kurtosis, underscores its empirical robustness and highlights its potential as a powerful tool for modeling highly heterogeneous environmental and meteorological phenomena.

	\subsection*{Acknowledgments}
	We gratefully acknowledge all data providers for their contributions, which have been appropriately cited in the statistical analysis. \\

	\appendix
	
	\section{Appendix A: Proof of CDF} \label{appendixA}
	
	\begin{theorem}
		The expression in Equation \eqref{ktcdf} is a CDF. 
	\end{theorem}
	\begin{proof} 
		The necessary and sufficient conditions for a function to be a CDF are:\\
		i) $\lim_{x \to 0} F(x) = 0$ and $\lim_{x \to \infty} F(x) = 1$.\\
		ii) $F(x)$ is a non-decreasing function of $x$.\\
		iii) $F(x)$ is right-continuous.\\
		i) To prove $\lim_{x \to 0} F(x) = 0$, we have
		\begin{align*}
			\lim_{x \to 0}  1-e^{\fx} &= 0 \quad \text{(Since Equation \eqref{td} is a CDF.)} \\
			\lim_{x \to 0} F(x) &= \lim_{x \to 0} 1-\left( 1-\left( 1-e^{\fx} \right)^\a \right)^\b \\
			&= 1-\left( 1-\left( 0 \right)^\a \right)^\b \\
			&= 0.
		\end{align*}
		To prove $\lim_{x \to \infty} F(x) = 1$, we have
		\begin{align*}
			\lim_{x \to \infty}  1-e^{\fx} &= 1 \quad \text{(Since Equation \eqref{td} is a CDF.)} \\
			\lim_{x \to \infty} F(x) &= \lim_{x \to \infty} 1-\left( 1-\left( 1-e^{\fx} \right)^\a \right)^\b \\
			&=  1-\left( 1-\left( 1 \right)^\a \right)^\b \\
			&= 1.
		\end{align*}
		ii) To prove $F(x)$ is a non-decreasing function of $x$, let's take $x_1 \leq x_2$. 
		\begin{align*}
			1-e^{\phi(\x_1;\th)} &\leq 1-e^{\phi(\x_2;\th)} \quad \text{(Since Equation \eqref{td} is a non-decreasing function.)}\\
			\text{Since, a>0 so,}\\
			\Rightarrow (1-e^{\phi(\x_1;\th)})^a &\leq (1-e^{\phi(\x_2;\th)})^a\\ 
			\Rightarrow 1-(1-e^{\phi(\x_1;\th)})^a &\geq 1-(1-e^{\phi(\x_2;\th)})^a\\
			\text{Since, b>0 so,}\\
			\Rightarrow (1-(1-e^{\phi(\x_1;\th)})^a)^b &\geq (1-(1-e^{\phi(\x_2;\th)})^a)^b\\
			\Rightarrow 1-(1-(1-e^{\phi(\x_1;\th)})^a)^b &\leq 1-(1-(1-e^{\phi(\x_2;\th)})^a)^b\\
			\Rightarrow F(x_1) &\leq F(x_2)\\
		\end{align*}
		iii) Being a CDF, \( G(x) = 1 - e^{\phi(x; \theta)} \) is continuous since the exponential function is a continuous function. Now, since the power function of a continuous function is also continuous, for \( a > 0 \), \( G(x)^a \) is also continuous. \\
		Similarly, for \( b > 0 \), the expression \( 1 - (1 - G(x)^a)^b \) is also continuous, as it is constructed by applying arithmetic and composition of continuous functions, which preserve continuity. \\
		Hence, \( F(x) \) is right-continuous, as required by the definition of a CDF.
	\end{proof}

	\section{Appendix B: $r^{th}$ Order Moments Calculation}\label{appendixB}
	\begin{proof}
		Using the series expansion of the PDF from Equation (\ref{exapnsionpdf}), we have:
		\begin{align}
			E(X^r) &= \int_0^\infty x^r f_X(x) \, dx \nonumber \\
			&= \sum_{i,j=0}^\infty \sum_{k=0}^i \eta_{i,j,k} \int_0^\infty x^r \theta (e^{\theta x} - 1) e^{\fx} \left(1 - e^{\fx}\right)^k \, dx,
		\end{align}
		where
		\[
		\eta_{i,j,k} = a b (-1)^{i+j+k} \binom{\b-1}{j} \binom{\a(j+1)-1}{i} \binom{i}{k}.
		\]
	Since,
		\[
		\left(1 - e^{\fx}\right)^k = \sum_{l=0}^k \binom{k}{l} (-1)^l e^{l \fx}.
		\]
		Thus,
		\[
		e^{\fx} \left(1 - e^{\fx}\right)^k = \sum_{l=0}^k \binom{k}{l} (-1)^l e^{(l+1) \fx}.
		\]
		Substitute \( \fx= \theta x - e^{\theta x} + 1 \):
		\[
		e^{(l+1) \fx} = e^{(l+1)(\theta x + 1)} e^{-(l+1) e^{\theta x}}.
		\]
		Therefore,
		\begin{align}
			E(X^r) &= \sum_{i,j=0}^\infty \sum_{k=0}^i \sum_{l=0}^k \eta_{i,j,k} \binom{k}{l} (-1)^l e^{l+1} \nonumber \\
			&\quad \times \int_0^\infty x^r \theta (e^{\theta x} - 1) e^{(l+1)\theta x} e^{-(l+1) e^{\theta x}} \, dx.
		\end{align}
		Let \( t = e^{\theta x} \), so \( x = \frac{\ln t}{\theta} \), \( dx = \frac{dt}{\theta t} \), and \( e^{\theta x} = t \). Then \( e^{(l+1)\theta x} = t^{l+1} \), and \( e^{-(l+1) e^{\theta x}} = e^{-(l+1) t} \). Also, \( e^{\theta x} - 1 = t - 1 \).
		\begin{align}
			I &= \int_0^\infty x^r \theta (e^{\theta x} - 1) e^{(l+1)\theta x} e^{-(l+1) e^{\theta x}} \, dx \nonumber \\
			&= \theta \int_{1}^{\infty} \left( \frac{\ln t}{\theta} \right)^r (t-1) t^{l+1} e^{-(l+1) t} \frac{dt}{\theta t} \nonumber \\
			&= \frac{1}{\theta^r} \int_1^\infty (\ln t)^r (t-1) t^{l} e^{-(l+1) t} \, dt \nonumber \\
			&= \frac{1}{\theta^r} \left[ \int_1^\infty (\ln t)^r t^{l+1} e^{-(l+1) t} \, dt - \int_1^\infty (\ln t)^r t^{l} e^{-(l+1) t} \, dt \right].
		\end{align}
		Define for integer \( r \ge 0 \):
		\[
		J(a,b) = \int_1^\infty (\ln t)^r t^{a-1} e^{-b t} \, dt.
		\]
		Then,
		\[
		I = \frac{1}{\theta^r} \left[ J(l+2, l+1) - J(l+1, l+1) \right].
		\]
		Now, using the identity \( (\ln t)^r = \left. \frac{d^r}{d\nu^r} t^\nu \right|_{\nu=0} \), we have:
		\[
		J(a,b) = \left. \frac{d^r}{d\nu^r} \int_1^\infty t^{\nu+a-1} e^{-b t} \, dt \right|_{\nu=0}.
		\]
		The inner integral evaluates to the upper incomplete Gamma function:
		\[
		\int_1^\infty t^{\nu+a-1} e^{-b t} \, dt = b^{-\nu-a} \Gamma(\nu+a, b),
		\]
		where \( \Gamma(s,b) = \int_b^\infty u^{s-1} e^{-u} \, du \).
		
		Hence,
		\[
		J(a,b) = \left. \frac{d^r}{d\nu^r} \left[ b^{-\nu-a} \Gamma(\nu+a, b) \right] \right|_{\nu=0}.
		\]
		Therefore, the final closed-form expression for the $r^{th}$ order moments is:
		\[
		E(X^r) = \sum_{i,j=0}^\infty \sum_{k=0}^i \sum_{l=0}^k \eta_{i,j,k} \binom{k}{l} (-1)^l e^{l+1} \frac{1}{\theta^r} \left[ J(l+2, l+1) - J(l+1, l+1) \right].
		\]
	\end{proof}
	
	\section{Appendix C: Moment Generating Function Calculation} \label{appendixC}
	\begin{proof}
		Using the series expansion of the PDF from Equation (\ref{exapnsionpdf}), the moment generating function is:
		\begin{align}
			M_X(t) &= E(e^{tX}) = \int_0^\infty e^{tx} f_X(x) dx \nonumber \\
			&= \sum_{i,j=0}^\infty \sum_{k=0}^i \eta_{i,j,k} \int_0^\infty e^{tx} \theta (e^{\theta x} - 1) e^{\fx} \left(1 - e^{\fx}\right)^k dx.
		\end{align}
		\begin{align}
			e^{\fx} \left(1 - e^{\fx}\right)^k &= \sum_{l=0}^k \binom{k}{l} (-1)^l e^{(l+1)\fx}.
		\end{align}
		Substituting \( \fx = \theta x - e^{\theta x} + 1 \):
		\begin{align}
			e^{(l+1) \fx} = e^{(l+1)(\theta x + 1)} e^{-(l+1) e^{\theta x}}.
		\end{align}
		Thus:
		\begin{align}
			M_X(t) &= \sum_{i,j=0}^\infty \sum_{k=0}^i \sum_{l=0}^k \eta_{i,j,k} \binom{k}{l} (-1)^l e^{l+1} \nonumber \\
			&\quad \times \int_0^\infty e^{tx} \theta (e^{\theta x} - 1) e^{(l+1)\theta x} e^{-(l+1) e^{\theta x}} dx.
		\end{align}
		Make the substitution \( u = e^{\theta x} \). Then \( x = \frac{\ln u}{\theta} \), \( dx = \frac{du}{\theta u} \), \( e^{tx} = u^{t/\theta} \), \( e^{(l+1)\theta x} = u^{l+1} \), and \( e^{\theta x} - 1 = u - 1 \). The integral becomes:
		\begin{align}
			&\int_0^\infty e^{tx} \theta (e^{\theta x} - 1) e^{(l+1)\theta x} e^{-(l+1) e^{\theta x}} dx \nonumber \\
			&= \theta \int_1^\infty u^{t/\theta} (u - 1) u^{l+1} e^{-(l+1)u} \frac{du}{\theta u} \nonumber \\
			&= \int_1^\infty u^{t/\theta + l} (u - 1) e^{-(l+1)u} du.
		\end{align}
		Separating into two integrals:
		\begin{align}
			\int_1^\infty u^{t/\theta + l} (u - 1) e^{-(l+1)u} du = \int_1^\infty u^{t/\theta + l + 1} e^{-(l+1)u} du - \int_1^\infty u^{t/\theta + l} e^{-(l+1)u} du.
		\end{align}
		The upper incomplete Gamma function: \( \Gamma(s, b) = \int_b^\infty z^{s-1} e^{-z} dz \). For \( b > 0 \),
		\begin{align}
			\int_1^\infty u^{a-1} e^{-b u} du = b^{-a} \Gamma(a, b).
		\end{align}
		Applying this with \( b = l+1 \):
		\begin{align}
			\int_1^\infty u^{t/\theta + l + 1} e^{-(l+1)u} du &= (l+1)^{-(t/\theta + l + 2)} \, \Gamma\!\left( \frac{t}{\theta} + l + 2, \, l+1 \right), \\
			\int_1^\infty u^{t/\theta + l} e^{-(l+1)u} du &= (l+1)^{-(t/\theta + l + 1)} \, \Gamma\!\left( \frac{t}{\theta} + l + 1, \, l+1 \right).
		\end{align}
		Therefore, the MGF is:
		\begin{align}
			M_X(t) &= \sum_{i,j=0}^\infty \sum_{k=0}^i \sum_{l=0}^k \eta_{i,j,k} \binom{k}{l} (-1)^l e^{l+1} \Bigg[ (l+1)^{-(\frac{t}{\theta} + l + 2)} \Gamma\!\left( \frac{t}{\theta} + l + 2, \, l+1 \right) \nonumber \\
			&\qquad - (l+1)^{-(\frac{t}{\theta} + l + 1)} \Gamma\!\left( \frac{t}{\theta} + l + 1, \, l+1 \right) \Bigg].
		\end{align}
	\end{proof}

	\section{Appendix D: Order statistics Calculation}\label{appendixD}
	The density function $f_{m:n}(x)$ of the $m^{th}$ order statistic, where $m = 1, \ldots, n$ derived from independent and identically distributed (i.i.d.) random variables $X_1, \ldots, X_n$ following KTD, can be expressed as:
	\begin{eqnarray*}
		f_{m:n}(x)=\frac{n!}{m!(n-m)!} \f(\x) \F^{m-1}(\x) \left(1-\F(x)\right)^{n-m}  
	\end{eqnarray*}
	\begin{align*}
		\F(\x)^{m-1} &= \left[ 1- (1-\F(\x))\right]^{m-1}\\
		&= \sum_{r=0}^{\infty} (-1)^r \binom{m-1}{r} (1-\F(\x))^r\\
		f_{m:n}(x)&=\frac{n!}{m!(n-m)!} \f(\x) \sum_{r=0}^{\infty} (-1)^r \binom{m-1}{r} (1-\F(\x))^{r+n-m}  \\
		f_{m:n}(x)&=\frac{n!}{m!(n-m)!} \a \b g(\x) G(\x)^{\a -1} \left( 1-G(\x)^{\a}\right)^{b-1}\sum_{r=0}^{\infty} (-1)^r \binom{m-1}{r}\left( 1-G(\x)^{\a}\right)^{b(r+n-m)}\\
		f_{m:n}(x)&=\frac{n!}{m!(n-m)!} \a \b g(\x) G(\x)^{\a -1} \sum_{r=0}^{\infty} (-1)^r \binom{m-1}{r} \left( 1-G(\x)^{\a}\right)^{b(r+n-m+1)-1}\\
		f_{m:n}(x)&= g(\x) \sum_{i,j,r=0}^{\infty}\sum_{k=0}^{i} \tau_{m:n} G(\x)^{k}\\
		f_{m:n}(x)&= \th (e^{\th \x} -1) e^{\fx} \sum_{i,j,r=0}^{\infty}\sum_{k=0}^{i} \tau_{m:n} \left(1-e^{\fx} \right)^{k}
	\end{align*}
	where $$\tau_{i,j,k,r}^{(m:n)} = \frac{n!}{(m-1)!(n-m)!} \a \b (-1)^{i+j+k+r}  \binom{m-1}{r} \binom{b(r+n-m+1)-1}{j} \binom{\a (j+1)-1}{i} \binom{i}{k} $$
	\bibliography{test}

\end{document}